\documentclass[journal]{IEEEtran}
\usepackage{epsfig,epic,wrapfig,cite}
\usepackage{color}
\usepackage{url}
\usepackage{bm}
\usepackage{amsmath,amssymb,graphicx}
\usepackage{cite}
\usepackage{graphics,psfrag,subfigure}
\newcommand{\Ex}{\textbf{\textsc{E}}}
\newtheorem{thm}{Theorem}

\newtheorem{prop}[thm]{Proposition}

\title{Throughput-Optimal Opportunistic Scheduling in the Presence of Flow-Level Dynamics\thanks{Research supported by NSF Grants 07-21286 and 08-31756, ARO MURI Subcontracts, and the DTRA grants HDTRA1-08-1-0016 and HDTRA1-09-1-0055.}\thanks{A shorter version of this paper appears in the Proc. IEEE INFOCOM 2010.}}
\author{\IEEEauthorblockN{Shihuan Liu\IEEEauthorrefmark{1}, {\em Student Member, IEEE,} Lei Ying\IEEEauthorrefmark{1}, {\em Member, IEEE,} and
R. Srikant \IEEEauthorrefmark{2}, {\em Fellow, IEEE}}\\
\IEEEauthorblockA{\IEEEauthorrefmark{1}Department of Electrical and Computer Engineering, Iowa State University\\
Email: \{liush08, leiying\}@iastate.edu}\\
\IEEEauthorblockA{\IEEEauthorrefmark{2}Department of Electrical and Computer Engineering, University of Illinois at Urbana-Champaign\\
Email: rsrikant@illinois.edu}}

\begin{document}
\maketitle
\vspace{-0.5in}
\begin{abstract}
We consider multiuser scheduling in wireless networks with channel variations and flow-level dynamics.  Recently, it has been shown that the MaxWeight algorithm, which is throughput-optimal in networks with a fixed number of users, fails to achieve the maximum throughput in the presence of flow-level dynamics. In this paper, we propose a new algorithm, called \emph{workload-based scheduling with learning}, which is provably throughput-optimal, requires no prior knowledge of channels and user demands, and performs significantly better than previously suggested algorithms.
\end{abstract}

\section{Introduction}
Multiuser scheduling is one of the core challenges in wireless communications. Due to channel fading and wireless interference, scheduling algorithms need to dynamically allocate resources based on both the demands of the users and the channel states to maximize network throughput.  The
celebrated MaxWeight algorithm developed in \cite{TasEph_92} for exploiting channel variations works as follows. Consider a network with a single base station and $n$ users, and further assume that the base station can transmit to only one user in each time slot. The MaxWeight algorithm computes the product of the queue length and current channel rate for each user, and chooses to transmit to that user which has the largest product; ties can be broken arbitrarily. The throughput-optimality property of the MaxWeight algorithm  was first established in \cite{TasEph_92}, and the results were later extended to more general channel and arrival models in \cite{AndKumRam_04,ErySriPer_05,neemodroh05}. The MaxWeight algorithm should be contrasted with other opportunistic scheduling such as \cite{liuchoshr01,tsevislar02} which exploit channel variations to allocate resources fairly assuming continuously backlogged users, but which are not throughput-optimal when the users are not continuously backlogged.

While the results in \cite{TasEph_92, AndKumRam_04, ErySriPer_05} demonstrate the power of  MaxWeight-based algorithms, they were obtained under the assumptions that \emph{the number of users in the network is fixed and the traffic flow generated by each user is long-lived, i.e., each user continually injects new bits into the network.} However, practical networks have flow-level dynamics: users arrive to transmit data and leave the network after the data are fully transmitted.  In a recent paper \cite{VanBorShn_09}, the authors show that the MaxWeight algorithm is in fact \emph{not throughput optimal} in networks with flow-level dynamics by providing a clever example showing the instability of the MaxWeight scheduling. The intuition is as follows: if a long-lived flow does not receive enough service, its backlog builds up, which forces the MaxWeight scheduler to allocate more service to the flow. This interaction between user backlogs and scheduling guarantees the correctness of the resource allocation. However, if a flow has only a finite number of bits,  its backlog does not build up over time and it is possible for the MaxWeight to stop serving such a flow and thus, the flow may stay in the network forever. Thus, in a network where finite-size flows continue to arrive, the number of flows in the network could increase to infinity. One may wonder why flow-level instability is important since, in real networks, base stations limit the number of simultaneously active flows in the network by rejecting new flows when the number of existing flows reaches a threshold. The reason is that, if a network model without such upper limits is unstable in the sense that the number of flows grows unbounded, then the corresponding real network with an upper limit on the number of flows will experience high flow blocking rates. This fact is demonstrated in our simulations later.

In \cite{VanBorShn_09}, the authors address this instability issue of MaxWeight-based algorithms, and establish necessary and sufficient conditions for the stability of networks with flow-level dynamics. The authors also propose throughput-optimal scheduling algorithms. However, as the authors mention in \cite{VanBorShn_09}, the proposed algorithms require prior knowledge of channel distribution and traffic distribution, which is difficult and sometimes impossible to obtain in practical systems, and further, the performance of the proposed algorithms is also not ideal. A delay-driven MaxWeight scheduler has also been proposed to stabilize the network under flow-level dynamics \cite{SadVec_09}. The algorithm however works only when the maximum achievable rates of the flows are identical.

Since flow arrivals and departures are common in reality, we are interested in developing practical scheduling algorithms that are throughput-optimal \emph{under flow-level dynamics.} We consider a wireless system with a single base station and multiple users (flows). The network contains  both long-lived flows, which keep injecting bits into the network, and short-lived flows, which have a finite number of bits to transmit. The main contributions of this paper include the following:
\begin{itemize}
\item We obtain the necessary conditions for flow-level stability of networks with both long-lived flows and short-lived flows. This generalizes the result in \cite{VanBorShn_09}, where only short-lived flows are considered.

\item We propose a simple algorithm for networks with short-lived flows only. Under this algorithm, each flow keeps track of the best channel condition that it has seen so far. Each flow whose current channel condition is equal to the best channel condition that it has seen during its lifetime is eligible for transmission. It is shown that an algorithm which uniformly and randomly chooses a flow from this set of eligible flows for transmission is throughput-optimal. Note that the algorithm is a purely opportunistic algorithm in that it selects users for transmission when they are in the best channel state that they have seen so far, without considering their backlogs.

\item Based on an optimization framework, we propose to use the estimated \emph{workload}, the number of time slots required to transmit the remainder of a flow based on the best channel condition seen by the flow so far, to measure the backlog of short-lived flows. By comparing this short-lived flow backlog to the queue lengths and channel conditions of the long-lived flows, we develop a new algorithm, named workload-based scheduling with learning, which is throughput-optimal under flow-level dynamics. The term "learning" refers to the fact that the algorithm learns the best channel condition for each short-lived flow and attempts to transmit when the channel condition is the best.

\item We use simulations to evaluate the performance of the proposed scheduling algorithm, and observe that the workload-based scheduling with learning performs significantly better than the MaxWeight scheduling in various settings.
\end{itemize}

The terminology of long-lived and short-lived flows above has to be interpreted carefully in practical situations. In practice, each flow has a finite size and thus, all flows eventually will leave the system if they receive sufficient service. Thus, all flows are short-lived flows in reality. Our results suggest that transmitting to users who are individually in their best estimated channel state so far is thus, throughput optimal. On the other hand, it is also well known that real network traffic consists of many flows with only a few packets and a few flows with a huge number of packets. If one considers the time scales required to serve the small-sized flows, the large-sized flows will appear to be long-lived (i.e., persistent forever) in the terminology above. Thus, if one is interested in performance over short time-scales, an algorithm which considers flows with a very large number of packets as being long-lived may lead to better performance and hence, we consider the more general model which consists of both short-lived flows and long-lived flows. Our simulations later confirm the fact that the algorithm which treats some flows are being long-lived leads to better performance although throughput-optimality does not require such a model. In addition, long-lived flows partially capture the scenario where all bits from a flow do not arrive at the base station all at once. This fact is also exploited in our simulation experiments.

\section{Basic Model}
%\begin{figure}[h]
%\begin{center}
%\includegraphics[width=3in]{Figures/model_instability.eps}
%\caption{The network model}
%\label{fig:model}
%\end{center}
%\end{figure}

{\bf Network Model:} We consider a discrete-time wireless downlink network with a single base station and many flows, each flow associates with a distinct mobile user. The base station can serve only one flow at a time.

{\bf Traffic Model:} The network consists of the following two types of flows:
\begin{itemize}
\item {\bf Long-lived flows:} Long-lived flows are traffic streams that are always in the network and continually generate bits to be transmitted.

\item {\bf Short-lived flows:} Short-lived flows are flows that have a finite number of bits to transmit. A short-lived flow enters the network at a certain time, and leaves the system after all bits are transmitted.
\end{itemize}
We assume that the set of long-lived flows is fixed, and short-lived flows arrive and depart. We let $l$ be the index for long-lived flows, $\cal L$ be the set of long-lived flows, and $L$ be the number of long-lived flows, i.e., $L=|{\cal L}|.$ Furthermore, we let $X_l(t)$ be the number of new bits injected by long-lived flow $l$ in time slot $t,$ where $X_l(t)$ is a discrete random variable with finite support, and independently and identically distributed (i.i.d.) across time slots. We also assume $\Ex[X_l(t)]=x_l$ and $X_l(t)\leq X^{\max}$ for all $l$ and $t.$

Similarly, we let $i$ be the index for short-lived flows, ${\cal I}(t)$ be the set of short-lived flows in the network at time $t,$ and $I(t)$ be the number of short-lived flows at time $t,$ i.e., $I(t)=|{\cal I}(t)|.$ We denote by $f_i$ the size (total number of bits) of short-lived flow $i,$ and  assume $f_i\leq F^{\max}$ for all $i.$

It is important to note that we allow different short-lived flows to have different maximum link rates. A careful consideration of our proofs will show the reader that the learning algorithm is not necessary if all users have the same maximum rate and that one can simply transmit to the user with the best channel state if it is assumed that all users have the same maximum rate. However, we do not believe that this is a very realistic scenario since SNR variations will dictate different maximum rates for different users.

{\bf Residual Size and Queue Length:} For a short-lived flow $i,$ let $Q_i(t)$ which we call the residual size, denote the number of bits still remaining in the system at time $t$. For a long-lived flow $l,$ let $Q_l(t)$ denote the number of bits stored at the queue at the base station.

{\bf Channel Model:} There is a wireless link between each user and the base station.  Denote by $R_i(t)$ the state of the link between short-lived flow $i$ and the base station at time $t$ (i.e., the maximum rate at which the base station can transmit to short-lived flow $i$ at time $t$), and $R_l(t)$ the state of the link between long-lived flow $l$ and the base station at time $t.$ We assume that $R_i(t)$ and $R_l(t)$ are discrete random variables with finite support. Define $R_i^{\max}$ and $R_l^{\max}$  to be the largest values that these random variables can take, i.e., $P(R_j(t) > R_j^{\max})=0$ for each $j\in {\cal L} \bigcup \left(\bigcup_t {\cal I}(t)\right).$ \begin{color}{black}Choose \end{color} $p^{\max}_s>0$ and $R^{\max}>0$ such that
\begin{eqnarray*}
&\Pr(R_i(t)=R_i^{\max})\geq p^{\max}_s &\forall i, t\\
&\max\left\{\max_{i} R^{\max}_i, \max_l R^{\max}_l\right\}\leq R^{\max}.&
\end{eqnarray*} The states of wireless links are assumed to be independent across flows and time
slots (but not necessarily identically distributed across flows). The independence assumption across time slots can be relaxed easily but at the cost of more complicated proofs.

\section{Workload-based Scheduling with Learning}
In this section, we introduce a new scheduling algorithm called Workload-based Scheduling with Learning (WSL).

\noindent{\bf Workload-based Scheduling with Learning:} For a short-lived flow $i,$ we define $$\tilde{R}_{i}^{\max}(t)=\max_{\max\{t-D, b_i\}\leq s \leq t} R_i(s),$$ where $b_i$ is the time short-lived flow $i$ joins the network and $D>0$ is called the learning period. A key component of this algorithm is to use $R^{\max}_i$ to evaluate the workload of short-lived flows (the reason will be explained in a detail in Section \ref{sec: th}). However, $R_i^{\max}$ is in general unknown, so the scheduling algorithm uses $\tilde{R}_i^{\max}(t)$ as an estimate of $R_i^{\max}.$

%We further define $\tilde{{\cal S}}(t)$ to be the following event: at time slot $t,$  there exists a short-lived flow (say flow $i$) satisfying $R_i(t)=\tilde{R}_i^{\max}(t)$ or $R_i(t)\geq Q_i(t).$

During each time slot, the base station first checks the following inequality:
\begin{eqnarray}
\alpha \sum_{i\in {\cal I}(t)} \left\lceil\frac{Q_i(t)}{\tilde{R}_i^{\max}(t)}\right\rceil > \max_{l\in{\cal L}} Q_l(t) R_l(t), \label{eq: svsl}
\end{eqnarray} where $\alpha>0.$
\begin{list}{\labelitemi}{\leftmargin=1em}%\renewcommand{\itemsep}{-1pt}
\item If inequality (\ref{eq: svsl}) holds, then the base station serves a short-lived flow as follows: if at least one short-lived flow (say flow $i$)  satisfies $R_i(t)\geq Q_i(t)$ or $R_i(t)=\tilde{R}^{\max}_i(t),$ then the base station selects such a flow for transmission (ties are broken according to a \emph{good} tie-breaking rule, which is defined at the end of this algorithm); otherwise, the base station picks an arbitrary short-lived flow to serve.

\item If inequality (\ref{eq: svsl}) does not hold, then the base station serves a long-lived flow $l^*$ such that $$l^*\in \arg\max_{l\in{\cal L}} Q_l(t) R_l(t)$$ (ties are broken arbitrarily).
\end{list}

{\bf ``Good'' tie-breaking rule:} Assume that the tie-breaking rule is applied to pick a short-lived flow every time slot (but the flow is served only if
$\alpha \sum_{i\in {\cal I}(t)} \left\lceil\frac{Q_i(t)}{\tilde{R}_i^{\max}(t)}\right\rceil > \max_{l\in{\cal L}} Q_l(t) R_l(t)$). We define ${\cal E}_{\scriptsize miss}(t)$ to be the event that  the tie-breaking rule selects a short-lived flow with $\tilde{R}^{\max}_i(t)\not=R^{\max}_i.$ \begin{color}{black} Define
$$W_s(t)=\sum_{i\in {\cal I}(t)} \left\lceil\frac{Q_i(t)}{{R}_i^{\max}}\right\rceil,$$
which is he total workload of the system at time $t.$\end{color} A tie-breaking rule is said to be \emph{good} if the following condition holds: Consider the WSL with the given tie-breaking rule and learning period $D.$ Given any $\epsilon_{\scriptsize miss}>0,$ there exist $N_{\epsilon_{\scriptsize miss}}$ and $D_{\epsilon_{\scriptsize miss}}$ such that $$\Pr\left({\cal E}_{\scriptsize miss}(t)\right)\leq \epsilon_{\scriptsize miss}$$ if $D\geq D_{\epsilon_{\scriptsize miss}}$ and $W_s(t-D)\geq N_{\epsilon_{\scriptsize miss}}.$
\rightline{$\square$}

\emph{Remark 1:} While all WSL scheduling algorithms with good tie-breaking rules are throughput optimal, their performances in terms of other metrics could be different depending upon the tie-breaking rules. We consider two tie-breaking rules in this paper:
\begin{itemize}
\item {\bf Uniform Tie-breaking:} Among all short-lived flows satisfying $R_i(t)=\tilde{R}_i^{\max}(t)$ or $R_i(t)\geq Q_i(t),$ the base-station uniformly and randomly selects one to serve.

\item {\bf Oldest-first Tie-breaking:} Let $\beta_i$ denote the number of time slots a short-lived flow has been in the network. The base station keeps track $\tau_i=\min\{\bar{\tau}, \beta_i\}$ for every short-lived flow,  where $\bar{\tau}$ is some fixed positive integer. Among all short-lived flows satisfying $R_i(t)=\tilde{R}_i^{\max}(t)$ or $R_i(t)\geq Q_i(t),$ the tie-breaking rule selects the one with the largest $\tau_i,$ and the ties are broken uniformly and randomly.\footnote{We set a upper bound $\bar{\tau}$ on $\beta$ for technical reasons that facilitate the throughput-optimality proof. Since $\bar{\tau}$ can be arbitrarily large, we conjecture that this upper bound is only for analysis purpose, and not required in practical systems.}
\end{itemize}
The ``goodness'' of these two tie-breaking rules are proved in Appendix C and D, and the impact of the tie-breaking rules on performance is studied in Section \ref{sec: simu} using simulations.

\emph{Remark 2:} The $\alpha$ in inequality (\ref{eq: svsl}) is a parameter balancing the performance of long-lived flows and short-lived flows. A large $\alpha$ will lead to a small number of short-lived flows but large queue-lengths of long-lived flows, and vice versa.

\emph{Remark 3:} In Theorem \ref{thm: learning}, we will prove that WSL is throughput optimal when $D$ is sufficiently large. From purely throughput-optimality considerations, it is then natural to choose $D=\infty.$ However, in practical systems, if we choose $D$ too large, such as $\infty,$ then it is possible that a flow may stay in the system for a very long time if its best channel condition occurs extremely rarely. Thus, it is perhaps best to choose a finite $D$ to tradeoff between performance and throughput.

\emph{Remark 4:} If all flows are short-lived, then the algorithm simplifies as follows:  If at least one short-lived flow (say flow $i$) satisfies $R_i(t)\geq Q_i(t)$ or $R_i(t)=\tilde{R}^{\max}_i(t),$ then the base station selects such a flow for transmission according to a ``good'' tie-breaking rule; otherwise, the base station picks an arbitrary short-lived flow to serve.  Simply stated, the algorithm serves one of the flows which can be completely transmitted or sees its best channel state, where the best channel state is an estimate based on past observations. If no such flow exists, any flow can be served. We do not separately prove the throughput optimality of this scenario since it is a special case of the scenario considered here. But it is useful to note that, in the case of short-lived flows only, the algorithm does not consider backlogs at all in making scheduling decisions.

We will prove that WSL (with any $\alpha>0$) is throughput-optimal in the following sections, i.e., the scheduling policy can support any set of traffic flows that are supportable by any other algorithm. In the next section, we first present the necessary conditions for the stability, which also define the network throughput region.

\section{Necessary Conditions for Stability}
\label{sec: NC}
In this section, we establish the necessary conditions for the stability of networks with flow-level dynamics. To get the necessary condition, we need to classify the short-lived flows into different classes.
\begin{list}{\labelitemi}{\leftmargin=1em}
\item A short-lived flow class is defined by a pair of random variables $(\hat{R},\hat{F})$. Class-$k$ is associated with random variables $\hat{R}_k$ and $\hat{F}_k.$\footnote{We use  $\hat{\hspace{0.1in}}$  to indicate that the notation is associated with a class of short-lived flows instead of an individual short-lived flow.}  A short-lived flow $i$ belongs to class $k$ if $R_i(t)$ has the same distribution as $\hat{R}_k$ and the size of flow $i$ ($f_i$) has the same distribution as $\hat{F}_k.$ We let $\Lambda_k(t)$ denote the number of class-$k$ flows joining the network at time $t,$ where $\Lambda_k(t)$ are i.i.d. across time slots \begin{color}{black}and independent but not necessarily identical across classes, \end{color}and $\Ex[\Lambda_k(t)]=\lambda_k.$ Denote by $\cal K$ the set of distinct classes. We assume that $\cal K$ is finite, $|{\cal K}|=K,$ and $\Lambda_k[t]\leq \lambda^{\max}$ for all $t$ and $k\in {\cal K}.$

\item Let $\mathbf{c}$ denote an $L$-dimensional vector describing the state of the channels of the long-lived flows. In state $\mathbf{c},$ $R_{\mathbf{c}, l}$ is the service rate that long-lived flow $l$ can receive if it is scheduled. We denote by ${\cal C}$ the set of all possible states.

\item Let ${\bf C}(t)$ denote the state of the long-lived flows at time $t,$ and $\pi_{\mathbf{c}}$ denote the probability that ${\bf C}(t)$  is in state $\mathbf{c}.$

\item Let $p_{\mathbf{c}, l}$ be the probability that the base station serves flow $l$ when the network is in state $\mathbf{c}.$ Clearly, for any $\mathbf{c},$ we have
$$\sum_{l\in{\cal L}} p_{\mathbf{c}, l}\leq 1.$$ Note that the sum could be less than $1$ if the base station schedules a short-lived flow in this state.

\item Let $\mu_{\mathbf{c},s}$ be the probability that the base station serves a short-lived flow when the network is in state $\mathbf{c}.$

\item Let $\Theta_{k, \beta}(t)$ denote the number of short-lived flows that belong to class-$k$ and have residual size $Q(t)=\beta.$ Note that $\beta$ can only take on a finite number of values.
\end{list}

\begin{thm}
Consider traffic parameters $\{x_l\}$ and $\{\lambda_k\},$ and suppose that there exists a scheduling policy guaranteeing $$\lim_{t\rightarrow\infty}\Ex\left[\sum_{l\in{\cal L}} Q_l(t)+ \sum_{k\in{\cal K}}\sum_{\beta=1}^{F^{\max}} \Theta_{k,\beta}(t)\right]<\infty.$$  Then there exist $p_{\mathbf{c},l}$ and $\mu_{\mathbf{c},s}$ such that the following inequalities hold:
\begin{eqnarray}
&\displaystyle x_l \leq \sum_{\mathbf{c}\in{\cal C}} \pi_{\mathbf{c}} p_{\mathbf{c},l}R_{\mathbf{c},l}\hspace{0.1in}\forall l\in{\cal L}.\label{NC: L}\\
&\displaystyle \sum_{k\in{\cal K}} \lambda_k {\bf E}\left[\left\lceil\frac{\hat{F}_k}{\hat{R}^{\max}_k}\right\rceil\right]\leq \sum_{\mathbf{c}\in{\cal C}}\mu_{\mathbf{c},s}\pi_{\mathbf{c}}.&\label{NC: S2}\\
&\displaystyle \left(\sum_{l\in{\cal L}}p_{\mathbf{c},l}\right)+\mu_{\mathbf{c},s}\leq 1 \hbox{ } \forall c\in{\cal C}. &\label{NC: S}
\end{eqnarray} \rightline{$\square$}
\label{thm: NC}
\end{thm}

Inequality (\ref{NC: L}) and (\ref{NC: S2}) state that the service allocated should be no less than the user requests if the flows are supportable. Inequality (\ref{NC: S}) states that the overall time used to serve long-lived and short-lived flows should be no more than the time available. \begin{color}{black} To prove this theorem, it can be shown that for any traffic for which we cannot find $p_{\mathbf{c},l}$ and $\mu_{\mathbf{c},s}$ satisfying the three inequalities in the theorem, a Lyapunov function can be constructed such that the expected drift of the Lyapunov function is larger than some positive constant under any scheduling algorithm, which implies the instability of the network. \end{color} The complete proof is based on the Strict Separation Theorem and is along the lines of a similar proof in \cite{ErySriPer_05}, \begin{color}{black}and is omitted in this paper. \end{color}

\section{Throughput Optimality of WSL}
\label{sec: th}

First, we provide some intuition into how one can derive the WSL algorithm from optimization decomposition considerations. Then, we will present our main throughput optimality results.
Given traffic parameters $\{x_l\}$ and $\{\lambda_k\},$ the necessary conditions for the supportability of the traffic is equivalent to the feasibility of the following constraints:
\begin{eqnarray}
& x_l\leq \sum_{\mathbf{c}\in{\cal C}} \pi_{\mathbf{c}} p_{\mathbf{c},l} R_{\mathbf{c},l}& \forall l\nonumber\\
&\sum_{k\in {\cal K}} \lambda_k \Ex\left[\left\lceil\frac{\hat{F}_k}{\hat{R}^{\max}_k}\right\rceil\right]\leq \sum_{\mathbf{c}\in{\cal C}} \mu_{\mathbf{c},s}\pi_{\mathbf{c}}\label{NC: SF}\\
&\sum_{l\in{\cal L}} p_{\mathbf{c},l}+\mu_{\mathbf{c},s}\leq 1 &\forall \mathbf{c}.\nonumber
\end{eqnarray}
For convenience, we view the feasibility problem as an optimization problem with the objective $\max A,$ where $A$ is some constant. While we have not explicitly stated that the $x$'s and $\mu$'s are non-negative, this is assumed throughout.

Partially augmenting the objective using Lagrange multipliers, we get
\begin{eqnarray*}
&\max A -\sum_{l\in{\cal L}} q_l(x_l-\sum_c \pi_{\mathbf{c}} p_{\mathbf{c},l}R_{\mathbf{c},l})-\\
&q_s\left(\sum_{k\in{\cal K}} \lambda_k \Ex\left[\left\lceil\frac{\hat{F}_k}{\hat{R}^{\max}_k}\right\rceil\right]- \sum_{\mathbf{c}\in{\cal C}} \mu_{\mathbf{c},s}\pi_{\mathbf{c}}\right)\\
s.t.&\sum_{l\in{\cal L}} p_{\mathbf{c},l}+\mu_{\mathbf{c},s}\leq 1 \hbox{ }\forall \mathbf{c}.
\end{eqnarray*} For the moment, let us assume Lagrange multipliers $q_l$ and $q_s$ are given. Then the maximization problem above can be decomposed into a collection of optimization problems, one for each $\mathbf{c}:$
\begin{eqnarray*}
&\displaystyle \max_{\begin{color}{black}p_{\mathbf{c},l}, \mu_{\mathbf{c},s}\end{color}} \sum_{l\in{\cal L}} q_l R_{\mathbf{c},l} p_{\mathbf{c},l} + q_s \mu_{\mathbf{c},s}\\
s.t.&\sum_{l\in{\cal L}} p_{\mathbf{c},l}+\mu_{\mathbf{c},s}\leq 1.
\end{eqnarray*}
It is easy to verify that one optimal solution to the optimization problem above is:
\begin{list}{\labelitemi}{\leftmargin=1em}
\item if $q_s>\max_{l\in{\cal L}} q_l R_{\mathbf{c},l},$ then $\mu_{\mathbf{c},s}=1$ and  $p_{\mathbf{c},l}=0 (\forall l);$
\item otherwise,
$\mu_{\mathbf{c},s}=0,$ and $p_{\mathbf{c},l^*}=1$ for some $l^*\in\arg\max q_l R_{\mathbf{c},l}$ and $p_{\mathbf{c},l}=0$   for other $l.$
\end{list}

The complementary slackness conditions give
$$q_l\left(x_l - \sum_{\mathbf{c}\in{\cal C}} \pi_{\mathbf{c}} p_{\mathbf{c},l} R_{\mathbf{c},l}\right)=0.$$ Since $x_l$ is the mean
arrival rate of long-lived flow $l$ and $\sum_{\mathbf{c}\in{\cal C}} \pi_{\mathbf{c}} p_{\mathbf{c},l} R_{\mathbf{c},l}$ is the mean service rate,
the condition on $q_l$ says that if the mean arrival rate is less than the mean service rate, $q_l$ is equal to zero. Along with the non-negativity condition on $q_l,$ this suggests that perhaps $q_l$ behaves likes a queue with these arrival and service rates. Indeed, it turns out that the mean of the queue lengths are proportional to Lagrange multipliers (see the surveys in \cite{linshrsri06,geoneetas06,ShaSri_07}). For long-lived flow $l,$ we can treat the queue-length $Q_l(t)$ as a time-varying estimate of Lagrange multiplier $q_l.$ Similarly $q_s$ can be associated with a queue whose arrival rate is $\sum_{k\in{\cal K}} \lambda_k \Ex\left[\left\lceil\frac{\hat{F}_k}{\hat{R}^{\max}_k}\right\rceil\right],$ which is the mean rate at which workload arrives where workload is
measured by the number of slots needed to serve a short-lived flow if it is served when its channel condition is the best. The service rate is $\sum_{\mathbf{c}\in{\cal C}} \mu_{\mathbf{c},s}\pi_{\mathbf{c}}$ which is the rate at which the workload can potentially decrease when a short-lived flow is picked for scheduling by the base station. Thus, the workload in the system can serve as a dynamic estimate of $q_s.$

Letting $\alpha W_s(t)$ ($\alpha>0$) be an estimate of $q_s,$ the observations above suggest the following workload-based scheduling algorithm if $R_i^{\max}$ are known.

\noindent{\bf Workload-based Scheduling (WS):}
%We define ${{\cal S}}(t)$ to be the event that such that at time $t,$ there exists $i\in {\cal I}(t)$ such that $R_i(t)={R}_i^{\max}$ or $R_i(t)\geq Q_i(t).$
During each time slot, the base station checks the following inequality:
\begin{eqnarray}
\alpha W_s(t) > \max_{l\in{\cal L}} Q_l(t) R_l(t). \label{eq: svsl1}
\end{eqnarray}
\begin{list}{\labelitemi}{\leftmargin=1em}
\item  If inequality (\ref{eq: svsl1}) holds, then the base station serves a short-lived flow as follows:  if at least one short-lived flow (say flow $i$) satisfies $R_i(t)\geq Q_i(t)$ or $R_i(t)={R}^{\max}_i,$ then such a flow is selected for transmission (ties are broken arbitrarily); otherwise, the base station picks an arbitrary short-lived flow to serve.

\item If inequality (\ref{eq: svsl1}) does not hold, then the base station serves a long-lived flow $l^*$ such that $l^*\in \arg\max_{l\in{\cal L}} Q_l(t) R_l(t)$ (ties are broken arbitrarily).

\item The factor $\alpha$ can be obtained from the optimization formulation by multiplying constraint (\ref{NC: SF}) by $\alpha$ on both sides
\end{list}
\rightline{$\square$}

However, this algorithm which was directly derived from dual decomposition considerations is not implementable since $R_i^{\max}$'s are unknown. So WSL uses $\tilde{R}_i^{\max}(t)$ to approximate $R_i^{\max}.$ Note that an inaccurate estimate of ${R}_i^{\max}$ not only affects the base station's decision on whether $R_i(t)=R_i^{\max},$ but also on its computation of $\left\lceil\frac{Q_i(t)}{{R}_i^{\max}}\right\rceil.$ However, it is not difficult to see that the error in the estimate of the total workload is a small fraction of the total workload when the total workload is large: when the workload is very large, the total number of short-lived flows is large since their file sizes are bounded. Since the arrival rate of short-lived flows is also bounded, this further implies that the majority of short-lived flows must have arrived a long time ago which means that with high probability, their estimate of their best channel condition must be correct.

Next we will prove that both WS and WSL can stabilize any traffic $x_l$ and $\lambda_k$ such that $(1+\epsilon)x_l$ and $(1+\epsilon)\lambda_k$ are {\em supportable,} i.e., satisfying the conditions presented in Theorem \ref{thm: NC}. In other words, the number of short-lived flows in the network and the queues for long-lived flows are all bounded. Even though WS is not practical, we study it first since the proof of its throughput optimality is easier and provides insight into the proof of throughput-optimality of WSL.

Let $$\mathbf{M}(t)=\left(\{Q_l(t)\}_{l\in{\cal L}}, \{\Theta_{k, \beta}(t)\}_{k\in{\cal K}, 1\leq \beta\leq F^{\max}}\right).$$ Since the base station makes decisions on $\mathbf{M}(t)$ and $\mathbf{R}(t)=\{\{R_i(t)\}_{i\in{\cal I}(t)}, \{R_l(t)\}_{l\in{\cal L}}\}$ under WS. It is easy to verify that $\mathbf{M}(t)$ is a \emph{finite-dimensional} Markov chain under WS. Assume that $\Lambda_k(t)$, $\hat{F}_k$ and $X_l(t)$ are such that the Markov chain $\mathbf{M}$ is \emph{irreducible} and \emph{aperiodic}.

\begin{thm}
Given any traffic $x_l$ and $\lambda_k$ such that $(1+\epsilon)x_l$ and $(1+\epsilon)\lambda_k$ are supportable, the Markov chain $\mathbf{M}(t)$ is \emph{positive-recurrent} under WS, and
$$\lim_{t\rightarrow\infty} \Ex\left[ \sum_{l\in {\cal L}} Q_l(t) + \sum_{i\in {\cal I}(t)} Q_i(t) \right] < \infty.$$
\label{thm: ori}
\end{thm}
\begin{proof}
We consider the following Lyapunov function:
\begin{equation}
V(t)=\alpha\left(W_s(t)\right)^2+\sum_{l\in{\cal L}}(Q_l(t))^2,
\end{equation}
and prove that
\begin{eqnarray*}
\Ex[V(t+1)-V(t)|{\bf M}(t)] \leq U_d1_{{\bf M}(t)\in\Upsilon} -\frac{\epsilon}{2}\left[\alpha \bar{\lambda}W_s(t) \right.\\
+\left. \sum_{l\in \cal{L}} Q_l(t)x_l \right] 1_{{\bf M}(t)\not\in\Upsilon}
\end{eqnarray*} for some $U_d>0,$ $\epsilon>0$, $\bar{\lambda}>0$, and a finite set $\Upsilon.$ Positive recurrence of $\mathbf{M}$ then follows from Foster's Criterion for Markov chains \cite{asm03}, and the boundedness of the first moment follows from \cite{MeyTwe_09}. The detailed proof is presented in Appendix A.

\end{proof}

We next study WSL, where $R_i^{\max}$ is estimated from the history. We define $\Theta_{k, \beta, r}(t)$ to be the number of short-lived flows that belong to class-$k,$ have a residual size of $\beta,$ and have $\tilde{R}_i^{\max}(t)=r.$ Furthermore, we define $$\tilde{\mathbf{M}}(n)=\left(\{Q_l(t)\}_{l\in{\cal L}}, \{\Theta_{k, \beta, r}(t)\}_{\substack{k\in{\cal K}\\ 1\leq \beta\leq F^{\max}\\ 1\leq r\leq \hat{R}^{\max}_k}}\right)_{(n-1)T+1\leq t\leq nT}$$ from some $T\geq D.$ It is easy to see that $\tilde{\mathbf{M}}(n)$ is a \emph{finite-dimensional} Markov chain under WSL.\footnote{This Markov chain is well-defined under the uniform tie-breaking rule. For other good tie-breaking rules, we may need to first slightly change the definition of $\tilde{M}(n)$ to include the information required for tie-breaking, and then use the analysis in Appendix B to prove the positive recurrence.}

\begin{thm}
Consider traffic $x_l$ and $\lambda_k$ such that $(1+\epsilon)x_l$ and $(1+\epsilon)\lambda_k$ are supportable. Given WSL with a \emph{good} tie-breaking rule,  there exists $D_\epsilon$ such that the Markov chain $\tilde{\mathbf{M}}(n)$ is \emph{positive-recurrent} under the WSL with learning period $D\geq D_\epsilon$ and the given tie-breaking rule. Further,
$$\lim_{t\rightarrow\infty} \Ex\left[ \sum_{l\in {\cal L}} Q_l(t) + \sum_{i\in {\cal I}(t)} Q_i(t) \right] < \infty.$$
\label{thm: learning}
\end{thm}
\begin{proof}
The proof of this theorem is built upon the following two facts:
\begin{list}{\labelitemi}{\leftmargin=1em}
\item When the number of short-lived flows is large, the majority of short-lived flows must have been in the network for a long time and have obtained the correct estimate of the best channel condition, which implies that $$\sum_{i\in {\cal I}(t)} \left\lceil\frac{Q_i(t)}{{R}_i^{\max}}\right\rceil \approx \sum_{i\in {\cal I}(t)} \left\lceil\frac{Q_i(t)}{\tilde{R}_i^{\max}(t)}\right\rceil.$$

\item When the number of short-lived flows is large, the short-lived flow selected by the base station (say flow $i$) has a high probability to satisfy $R_i(t)=R_i^{\max}$ or $R_i(t)\geq Q_i(t).$
\end{list}

From these two facts, we can prove that with a high probability, the scheduling decisions of WSL are the same as those of WS, which leads to the throughput optimality of WSL. The detailed proof is presented in Appendix B.

\end{proof}

\section{Simulations}
\label{sec: simu}

In this section, we use simulations to evaluate the performance of different variants of WSL and compare it to other scheduling policies. There are three types of flows used in the simulations:
\begin{list}{\labelitemi}{\leftmargin=1em}
\item {\bf S-flow: } An S-flow has a finite size, generated from a truncated exponential distribution with mean value $30$ and maximum value $150.$ Non-integer values are rounded to integers.

\item {\bf M-flow: } An M-flow keeps injecting bits into the network for $10, 000$ time slots and stops.  The number of bits generated at each time slot follows a Poisson distribution with mean value $1.$

\item {\bf L-flow:} An L-flow keeps injecting bits into the network and never leaves the network. The number of bits generated at each time slot follows a truncated Poisson distribution with mean value $1$ \begin{color}{black}and maximum value $10$.\end{color}
\end{list}
Here S-flows represent short-lived flows that have finite sizes and whose bits arrive all at once; L-flows represent long-lived flows that continuously inject bits and never leave the network; and M-flows represent flows of finite size but whose arrival rate is controlled at their sources so that they do not arrive instantaneously into the network. Our simulation will demonstrate the importance of modeling very large, but finite-sized flows as long-lived flows.

We assume that the channel between each user and the base station is distributed according to one of the following three distributions:
\begin{list}{\labelitemi}{\leftmargin=1em}
\item {\bf G-link: } A G-link has five possible link rates $\{10, 20, 30, 40, 50\},$ and each of the states happens with probability $20\%.$

\item {\bf P-link: } A P-link has five possible link rates $\{5, 10, 15, 20, 25\},$ and each of the states happens with probability $20\%.$

\item {\bf R-link:} An R-link has five possible link rates $\{10, 20, 30, 40, 100\},$ and the probabilities associated with these link states are $\{0.5, 0.2, 0.2, 0.09, 0.01\}.$
\end{list}
The G, P and R stand for Good, Poor and Rare, respectively. We include these three different distributions to model the SNR variations among the users, where G-links represent links with high SNR (e.g., those users close to the base station), P-links represent links with low SNR (e.g., those users far away from the base station), and R-links represent links whose best state happens rarely. The R-links will be used to study the impact of learning period $D$ on the network performance.

We name the WSL with the uniform tie-breaking rule WSLU, and the WSL with the oldest-first tie-breaking rule WSLO. In the following simulations, we will first demonstrate that the WSLU performs significantly better than previously suggested algorithms, and then show that the performance can be further improved by choosing a good tie-breaking policy (e.g., WSLO).  We set $\alpha$ to be $50$ in all the following simulations.

\subsection*{Simulation I: Short-lived Flow or Long-lived Flow?}
We first use the simulation to demonstrate the importance of considering a flow with a large number of packets as being long-lived. We consider a network consisting of multiple S-flows and three M-flows, where the arrival of S-flows follows a truncated Poisson process with maximum value $100$ and mean value $\lambda.$ All the links are assumed to be G-links. We evaluate the following two schemes:
\begin{list}{\labelitemi}{\leftmargin=1em}
\item {\bf Scheme-1:} Both S-flows and M-flows are considered to be short-lived flows.

 \item {\bf Scheme-2:} An M-flow is considered to be long-lived before its last packet arrives, and to be short-lived after that.
\end{list}

The performance of these two schemes are shown in Figure \ref{fig: switch}, where WS with Uniform Tie-breaking Rule is used as the scheduling algorithm. We can see that the performances are substantially different (note that the network is stable under both schemes). The number of queued bits of M-flows under Scheme-1 is larger than that under Scheme-2 by \emph{two orders of magnitude.} This is because even an M-flow contains a huge number of bits ($10,000$ on average), it can be served only when the link rate is $50$ under Scheme-1.  This simulation suggests that when the performance we are interested is at a small scale (e.g. acceptable queue-length being less than or equal to $100$) compared with the size of the flow (e.g., $10^4$ in this simulation), the flow should be viewed as a long-lived flow for performance purpose.

\begin{figure}[h]
\begin{center}
\includegraphics[width=2.6in]{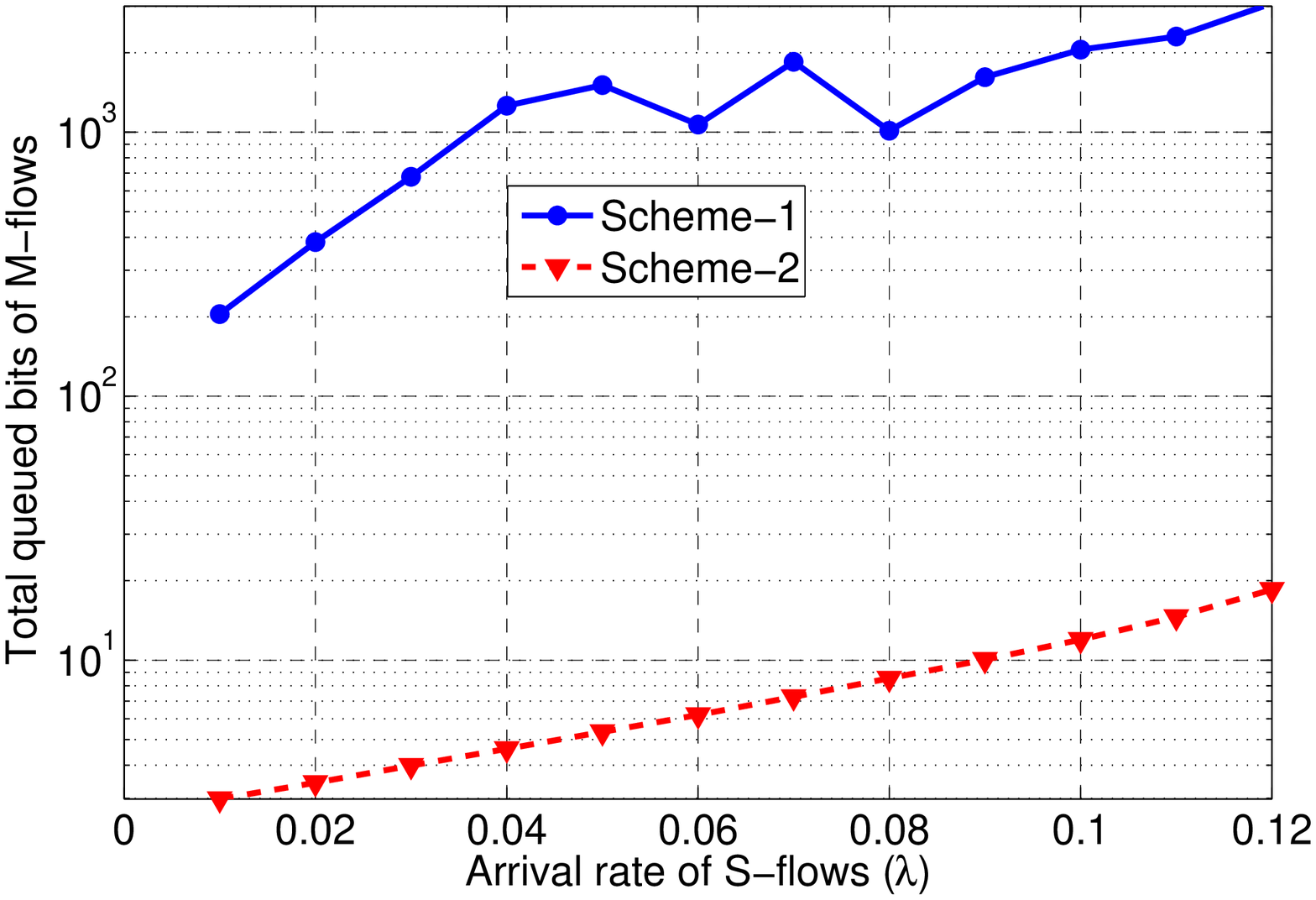}
\includegraphics[width=2.6in]{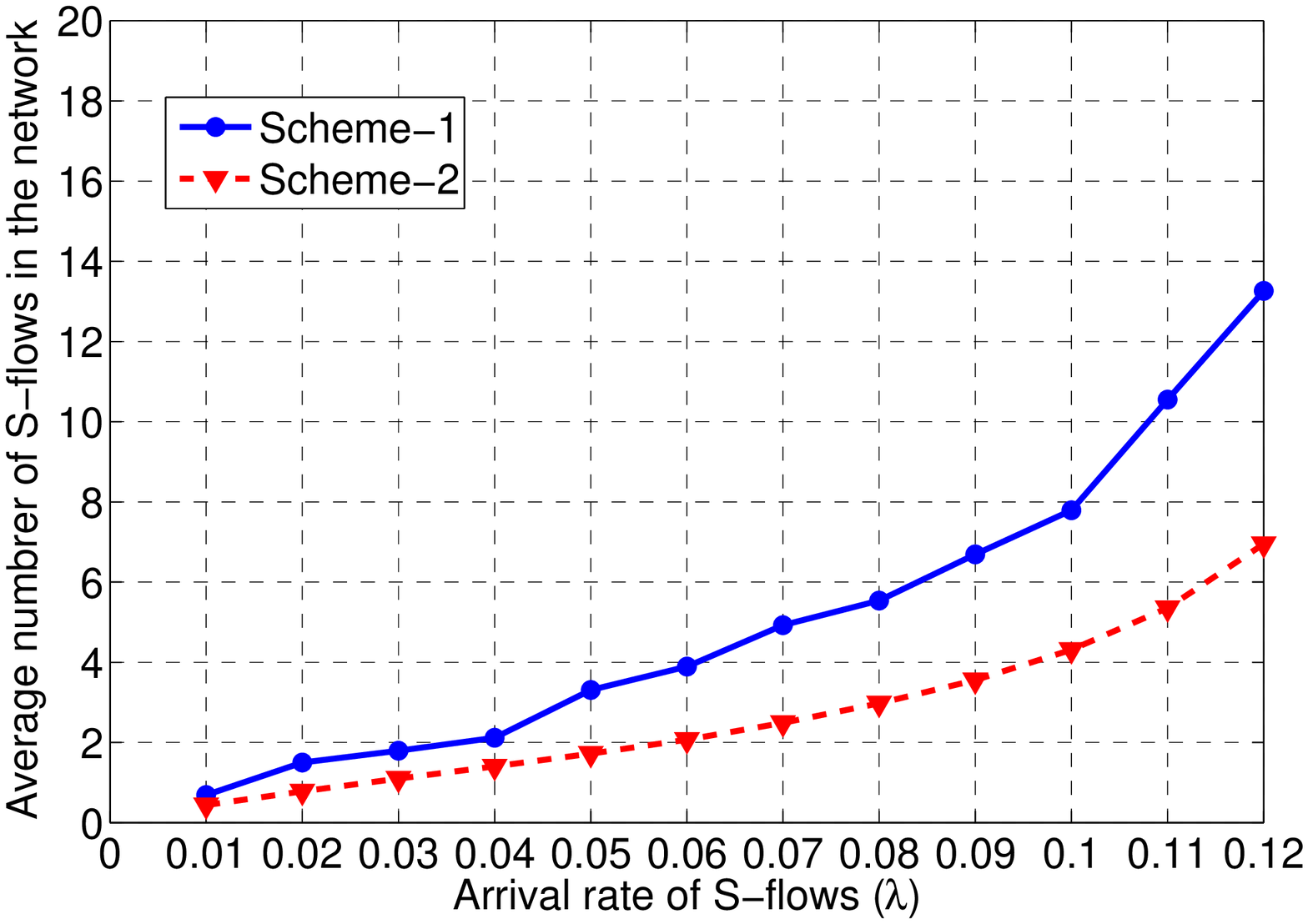}
\end{center}
\caption{Scheme-1 treats M-flows as short-lived flows, and Scheme-2 treats M-flows as long-lived flows}
\label{fig: switch}
\end{figure}

\subsection*{Simulation II: The Impact of Learning Period $D$}
In this simulation, we investigate the impact of $D$ on the performance of WSLU. Recall that it is nature to choose $D=\infty$ for purely throughput-optimality considerations, but the disadvantage is that a flow may stay in the network for a very long time if the best link state occurs very rarely. We consider a network consisting of S-flows, which arrive according to a truncated Poisson process with maximum value $100$ and mean $\lambda,$ and three L-flows. All links are assumed to be R-links. Figure \ref{fig: D} depicts the mean and standard deviation of the file-transfer delays with $D=16$ and $D=\infty$ \begin{color}{black} when the traffic load is light or medium.\end{color} As we expected, the standard deviation under WSLU with $D=\infty$ is significantly larger than that under WSLU with $D=16$ when $\lambda$ is large. This occurs because the best link rate $100$ occurs with a probability $0.01.$ This simulation confirms that in practical systems, we may want to choose a finite $D$ to get desired performance.
\begin{figure}[h]
\begin{center}
\includegraphics[width=2.6in]{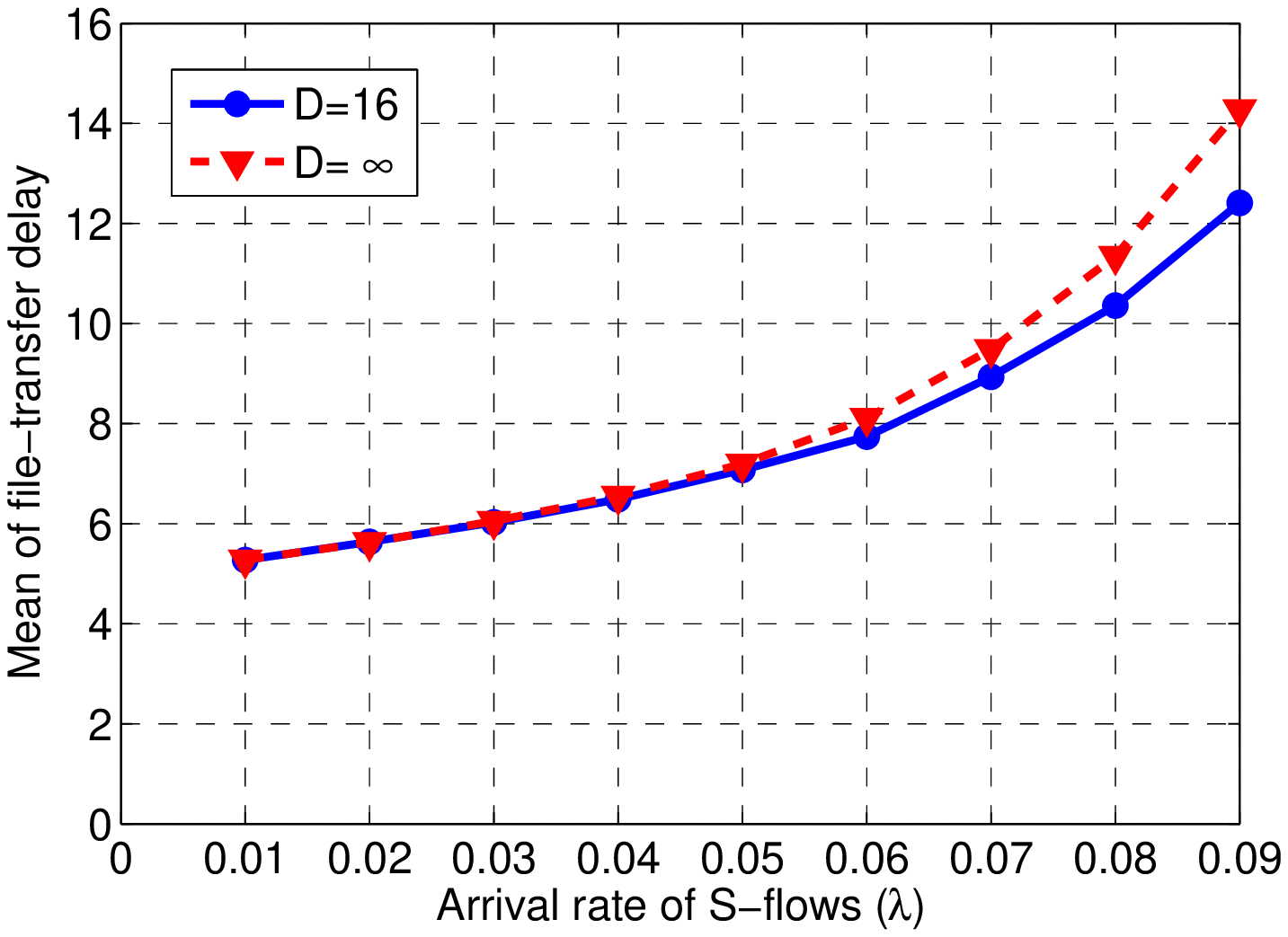}
\includegraphics[width=2.6in]{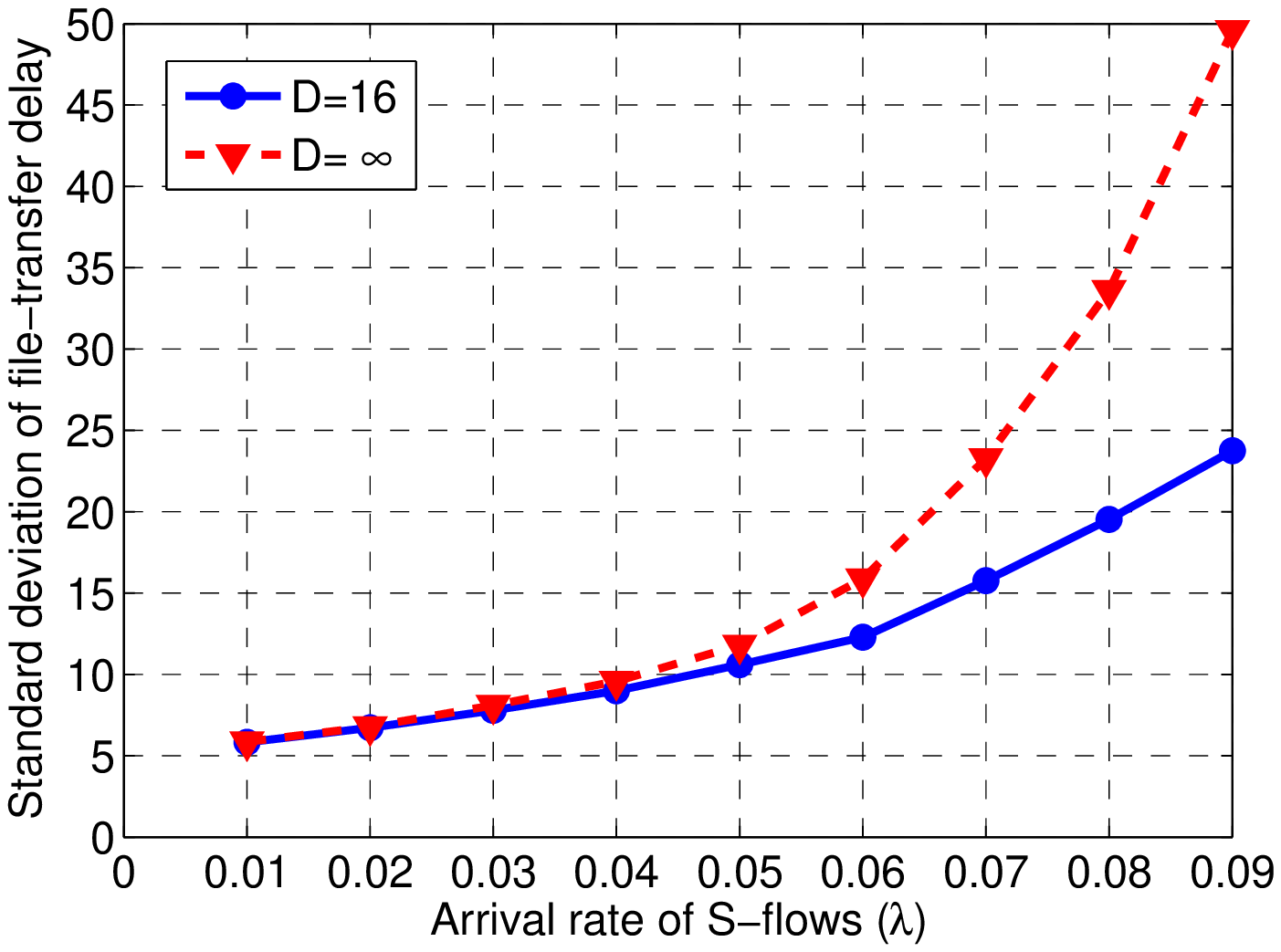}
\caption{The performance of WSLU with $D=16$ and $D=\infty$ when the traffic load is light or medium}
\label{fig: D}
\end{center}
\end{figure}

\begin{color}{black}
Further we would like to comment that while the WSLU algorithm with a small $D$ has a better performance in light or medium traffic regimes, throughput optimality is only guaranteed when $D$ is sufficiently large. Figure \ref{fig: D_heavy} illustrates the average number of S-flows and average file-transfer delay for $D=16$ and $D=\infty$ in heavy traffic regime. We can observe that in the heavy traffic regime, the WSLU with $D=\infty$ still stabilizes the network but the algorithm with $D=16$ does not. So there is a clear tradeoff in choosing $D$: A small $D$ reduces the file-transfer delay in light or medium traffic regimes, but a large $D$ guarantees stability in heavy traffic regime.
\end{color}

\begin{figure}[h]
\begin{center}
\includegraphics[width=2.6in]{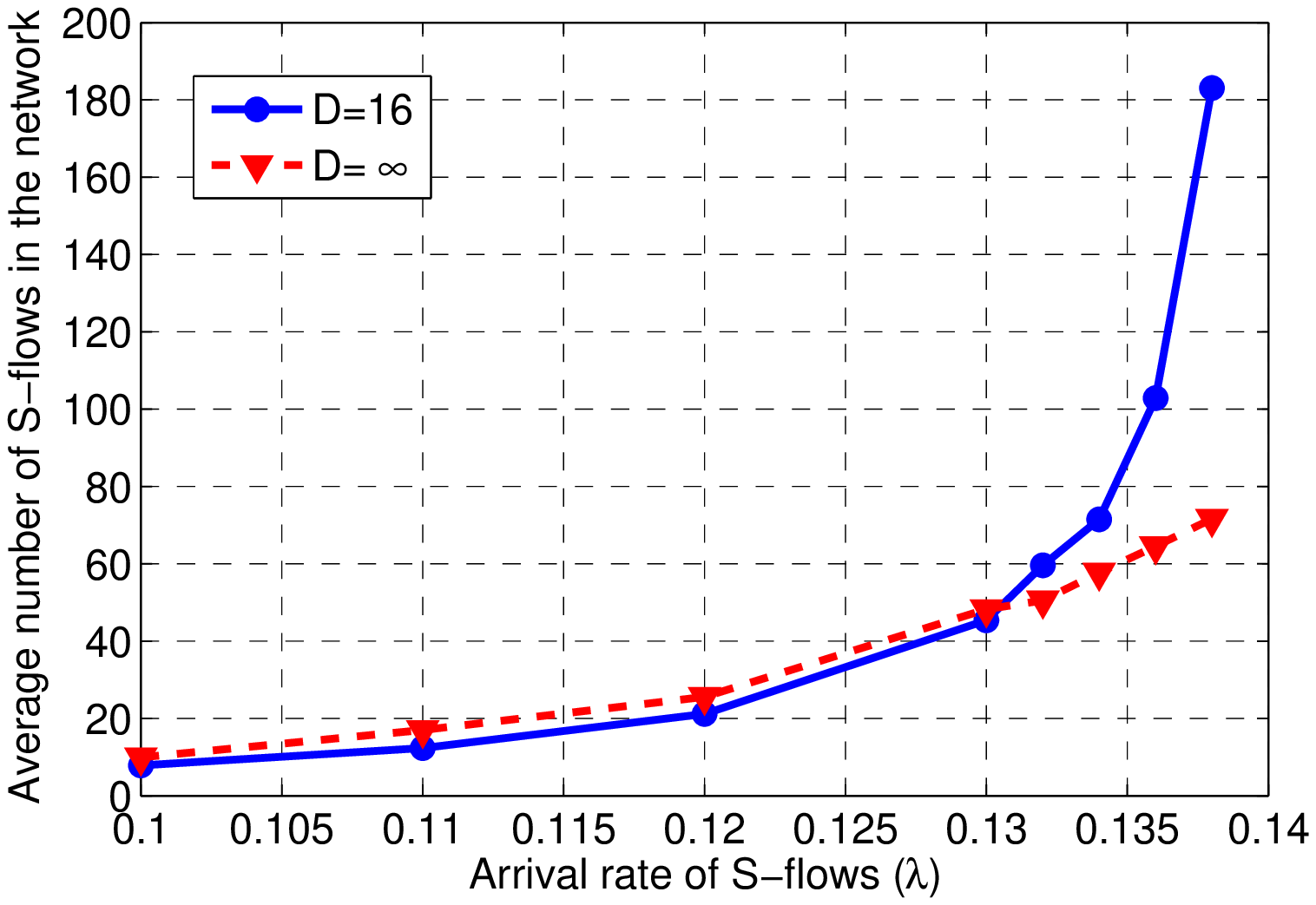}
\includegraphics[width=2.6in]{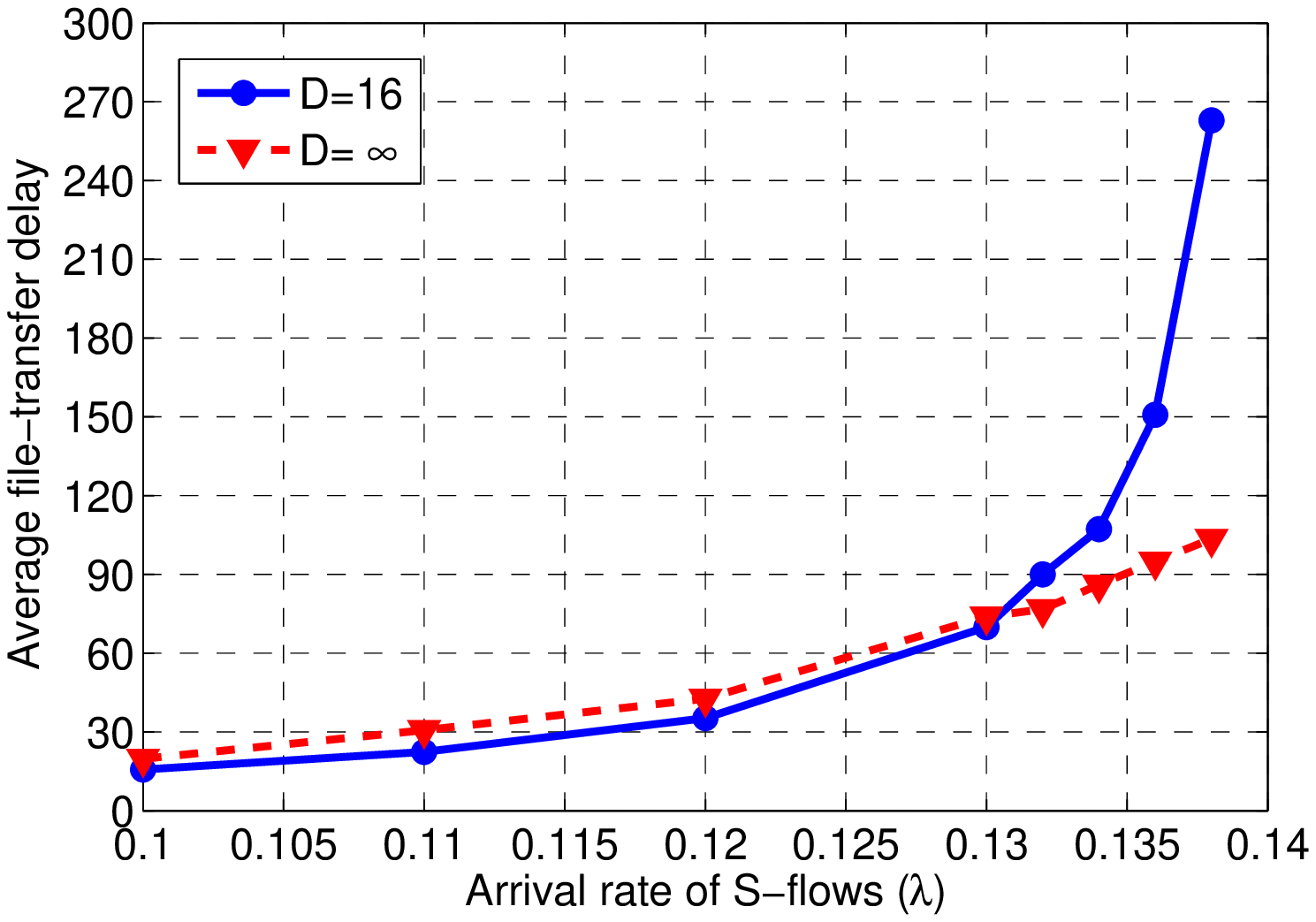}
\caption{The performance of WSLU with $D=16$ and $D=\infty$ when the traffic load is heavy}
\label{fig: D_heavy}
\end{center}
\end{figure}

\subsection*{Simulation III: Performance comparison of various algorithms}
\begin{color}{black}
In the following simulations, we choose $D=16.$  In the introduction, we have pointed out that the MaxWeight is not throughput optimal under flow-level dynamics because the backlog of a short-lived queue does not build up even when it has not been served for a while. To overcome this, one could try to use the delay of the head-of-line packet, instead of queue-length, as the weight because the head-of-line delay will keep increasing if no service is received. In the case of long-lived flows only, this algorithm is known to be throughput-optimal \cite{ErySriPer_05}. We will show that this Delay-based scheduling does not solve the instability problem when there are short-lived flows.

{\bf Delay-based Scheduling:} At each time slot, the base station selects a flow $i$ such that  $i\in \arg\max_i D_i(t)R_i(t),$ where $D_i(t)$ is the delay experienced so far by the head-of-line packet of flow $i.$
\end{color}

We first consider the case where all flows are S-flows, which arrive according to a truncated Poisson process with maximum value $100$ and mean $\lambda.$  An S-flow is assigned with a G-link or a P-link equally likely.

Figure \ref{fig: DSLWO} shows the average file-transfer delay  and average number of S-flows under different values of $\lambda.$  We can see that WSLU performs significantly better than the MaxWeight and Delay-based algorithms. Specifically, under MaxWeight and Delay-based algorithms, both the number of S-flows and file-transfer delay explode when $\lambda\geq 0.102.$ WSLU, on the other hand, performs well even when $\lambda=0.12.$

%Furthermore, when compared with WSLU,  WSLP reduces the number of S-flows and the file-transfer delay approximately $40\%$ when  $\lambda=0.12,$ which indicates the importance of selecting a good tie-breaking rule for improving the network performance.

\begin{figure}[h]
\begin{center}
\includegraphics[width=2.6in]{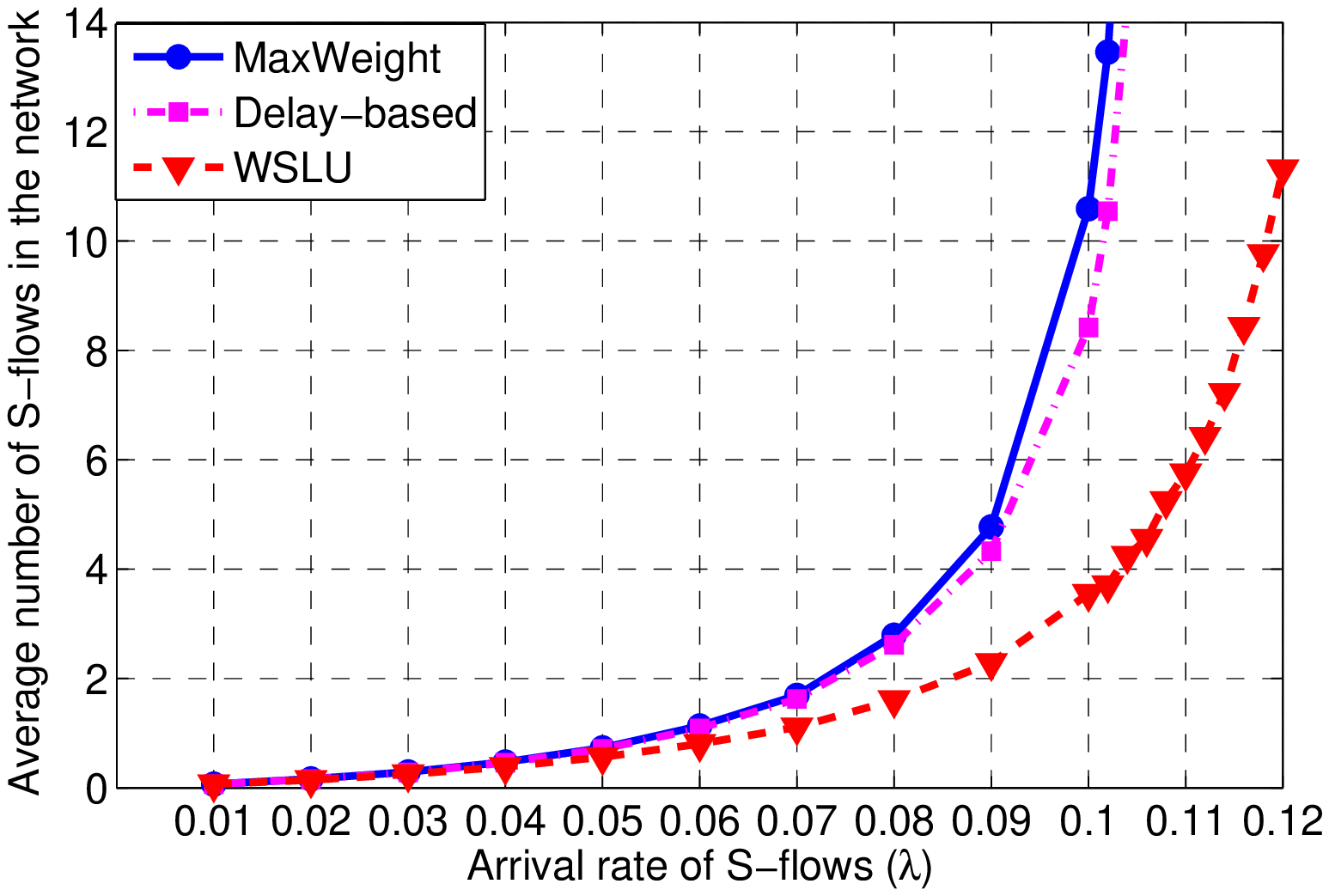}
\includegraphics[width=2.6in]{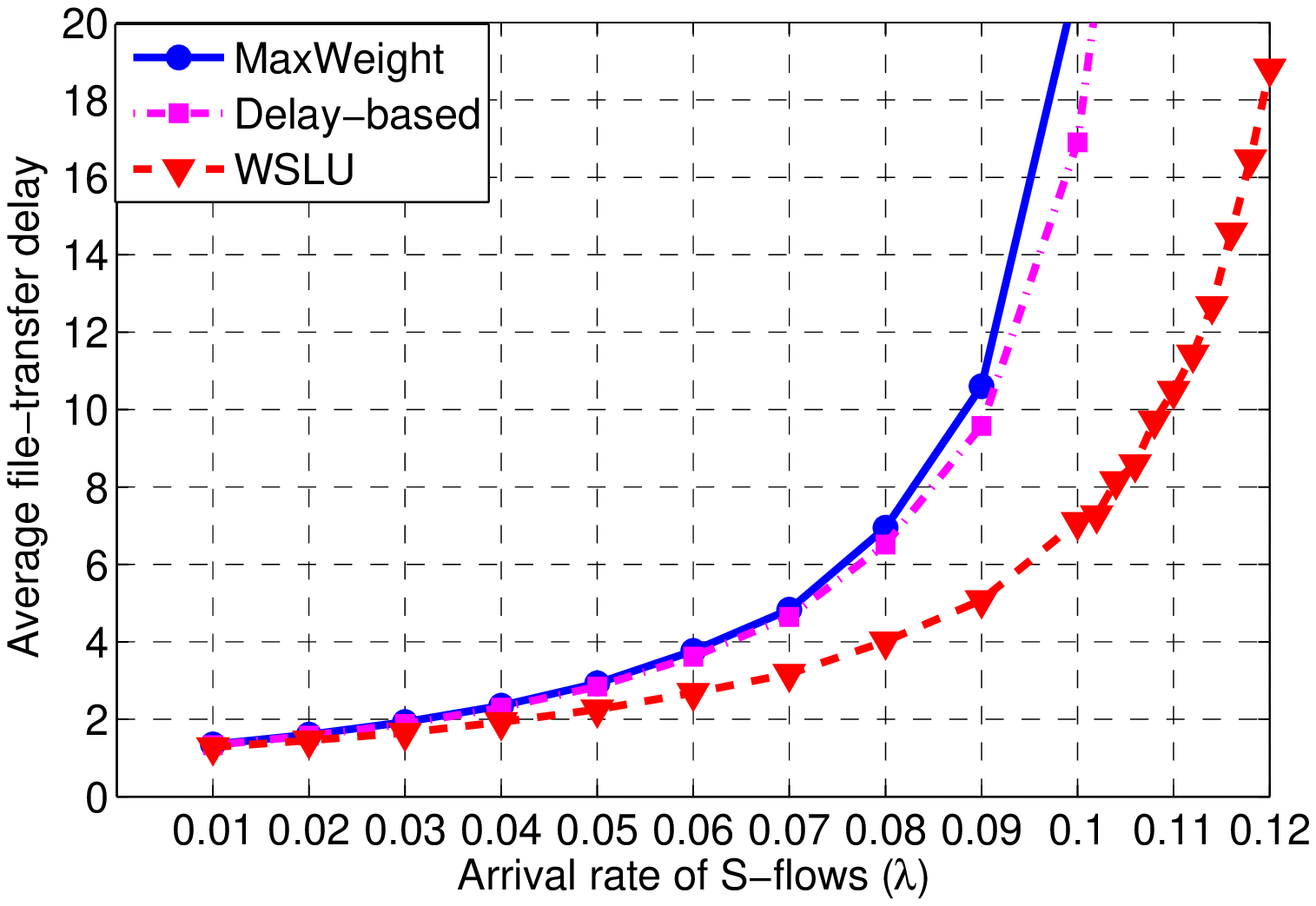}
\end{center}
\caption{The performance of the Delay-based, MaxWeight, and WSLU algorithms in a network without L-flows}
\label{fig: DSLWO}
\end{figure}

Next, we consider the same scenario with three L-flows in the network. Two of the L-flows have G-links and one has a P-link. Figure \ref{fig: SLW} shows the average number of short-lived flows and average file-transfer delay under different values of $\lambda.$ We can see that the MaxWeight becomes unstable even when \emph{the arrival rate of S-flows is very small.}  This is because the MaxWeight stops serving S-flows when the backlogs of L-flows are large, so S-flows stay in the network forever. The delay-based scheduling performs better than the MaxWeight, but significantly worse than WSLU.
\begin{figure}[h]
\begin{center}
\includegraphics[width=2.6in]{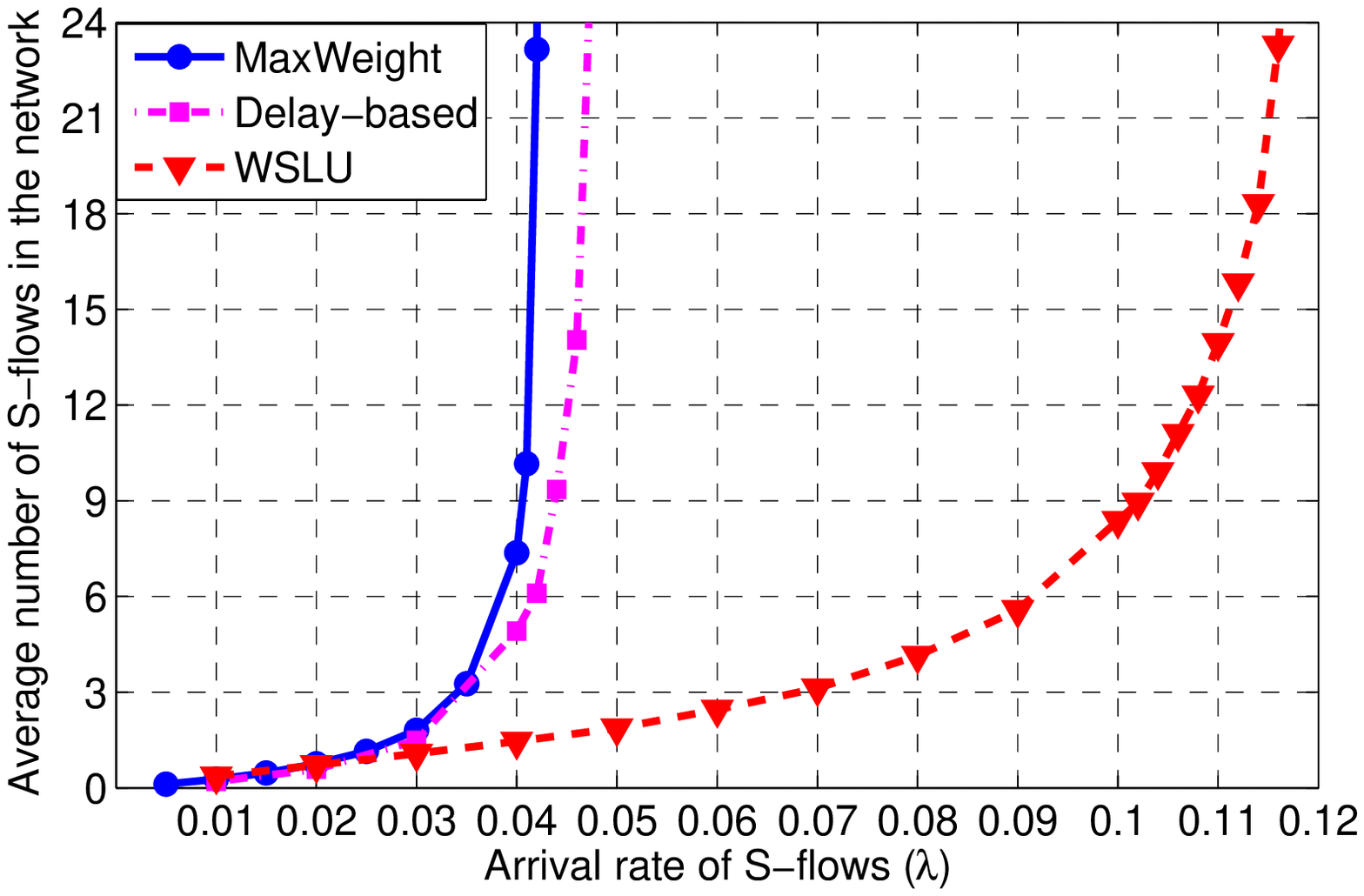}
\includegraphics[width=2.6in]{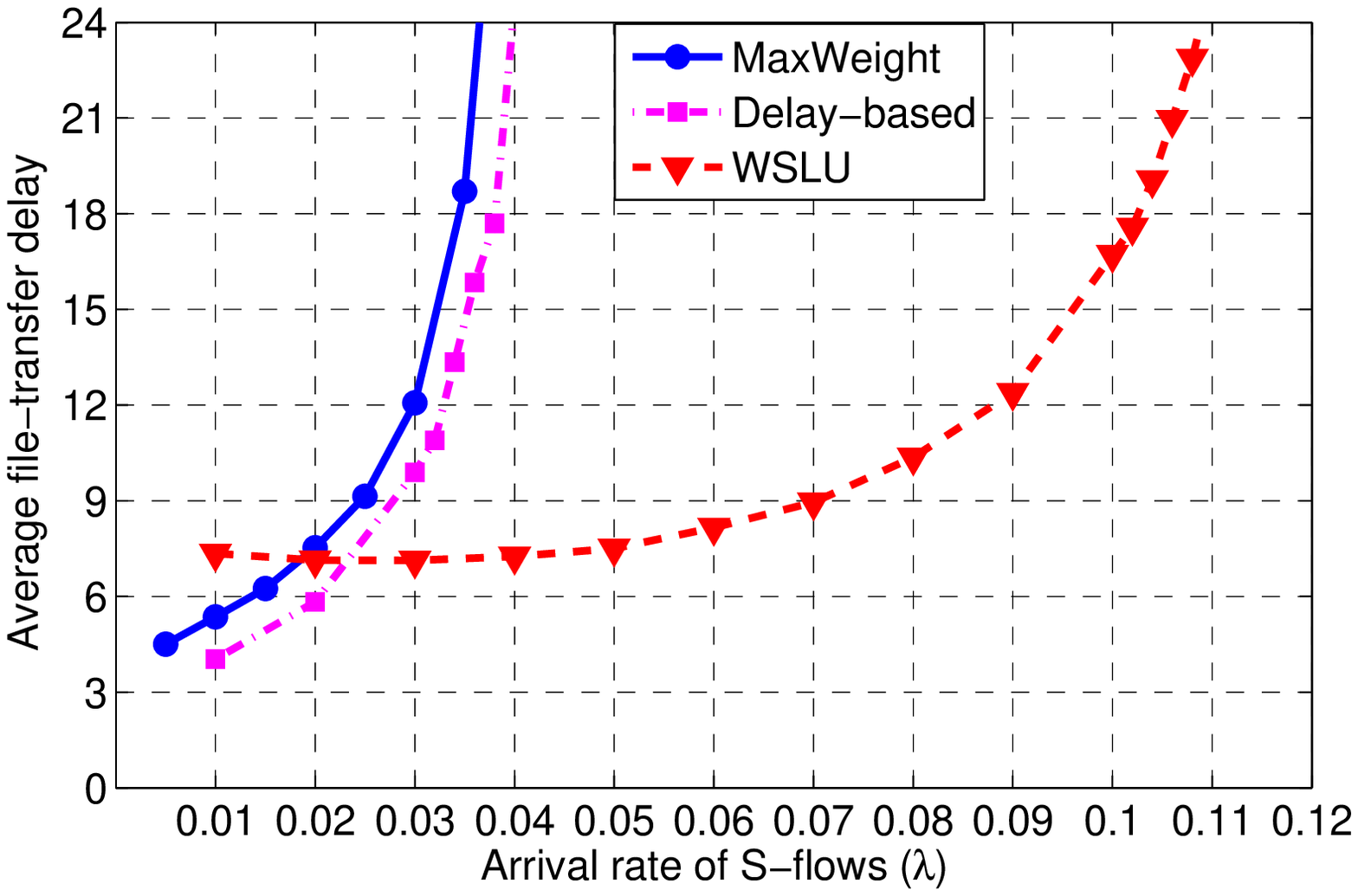}
\end{center}
\caption{The performance of the Delay-based, MaxWeight, and WSLU algorithms in a network with both S-flows and L-flows}
\label{fig: SLW}
\end{figure}

\subsection*{Simulation IV: Blocking probability of various algorithms}
While our theory assumes that the number of flows in the network can be infinite, in reality, base stations limit the number of simultaneously active flows, and reject new flows when the number of existing flows above some threshold. In this simulation, we assume that the base station can support at most $20$ S-flows. A new S-flow will be blocked if $20$ S-flows are already in the network. In this setting, the number of flows in the network is finite, so we compute the blocking probability, i.e., the fraction of S-flows rejected by the base station.

We consider the case where no long-lived flow is in the network and the case where both short-lived and long-lived flows are present in the network. The flows and channels are selected as in Simulation III.  The results are shown in Figure \ref{fig: Blocking_1} and \ref{fig: Blocking_2}. We can see that the blocking probability under WSLU is substantially smaller than that under the MaxWeight or the delay-based scheduling. Thus, this simulation demonstrates that instability under the assumption when the number of flows is allowed to unbounded implies high blocking probabilities for the practical scenario when the base station limits the number of flows in the network.

\begin{figure}[h]
\begin{center}
\includegraphics[width=2.6in]{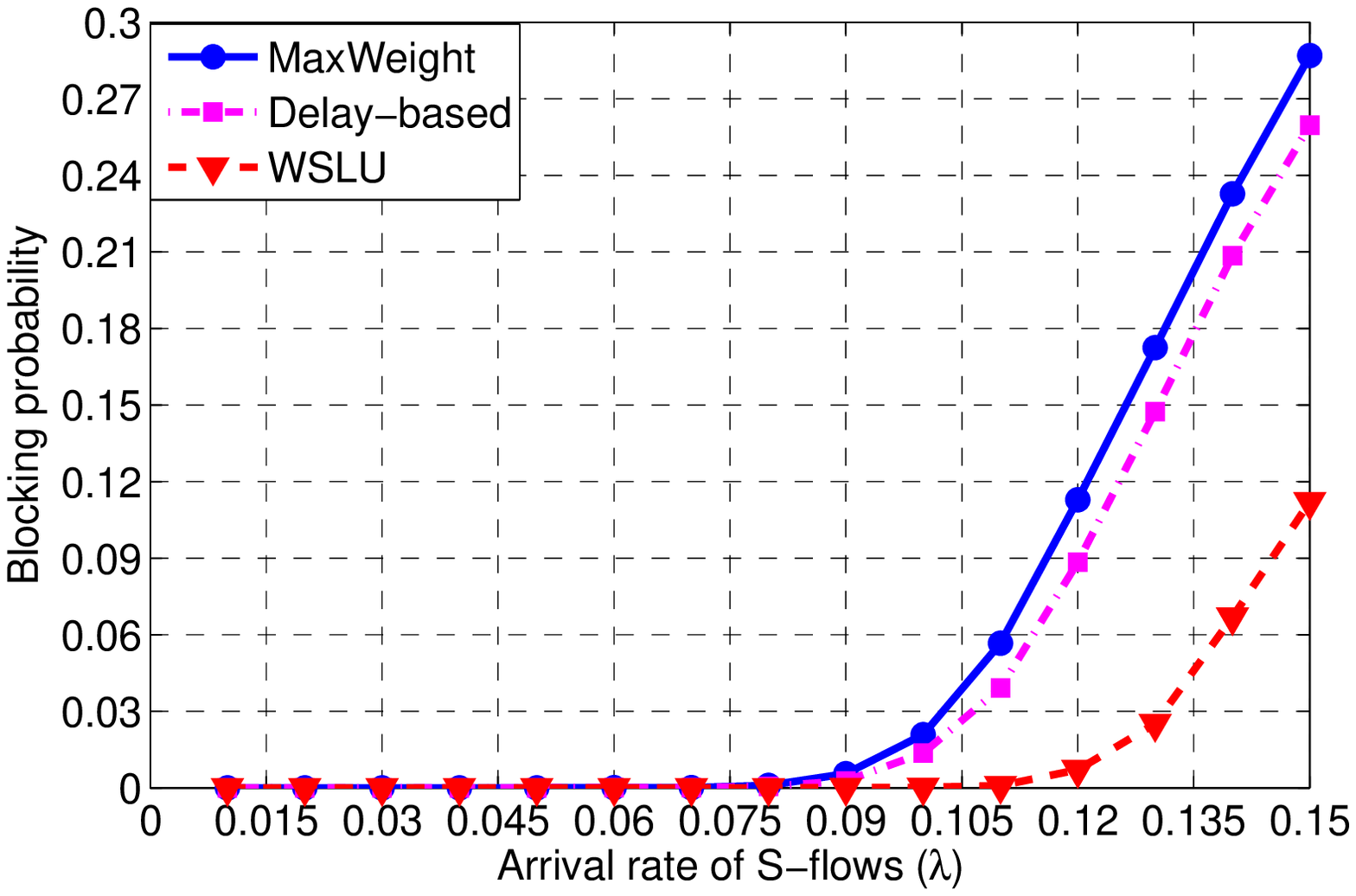}
\caption{The blocking probabilities of the Delay-based, MaxWeight, and WSLU  in a network without L-flows}
\label{fig: Blocking_1}
\includegraphics[width=2.6in]{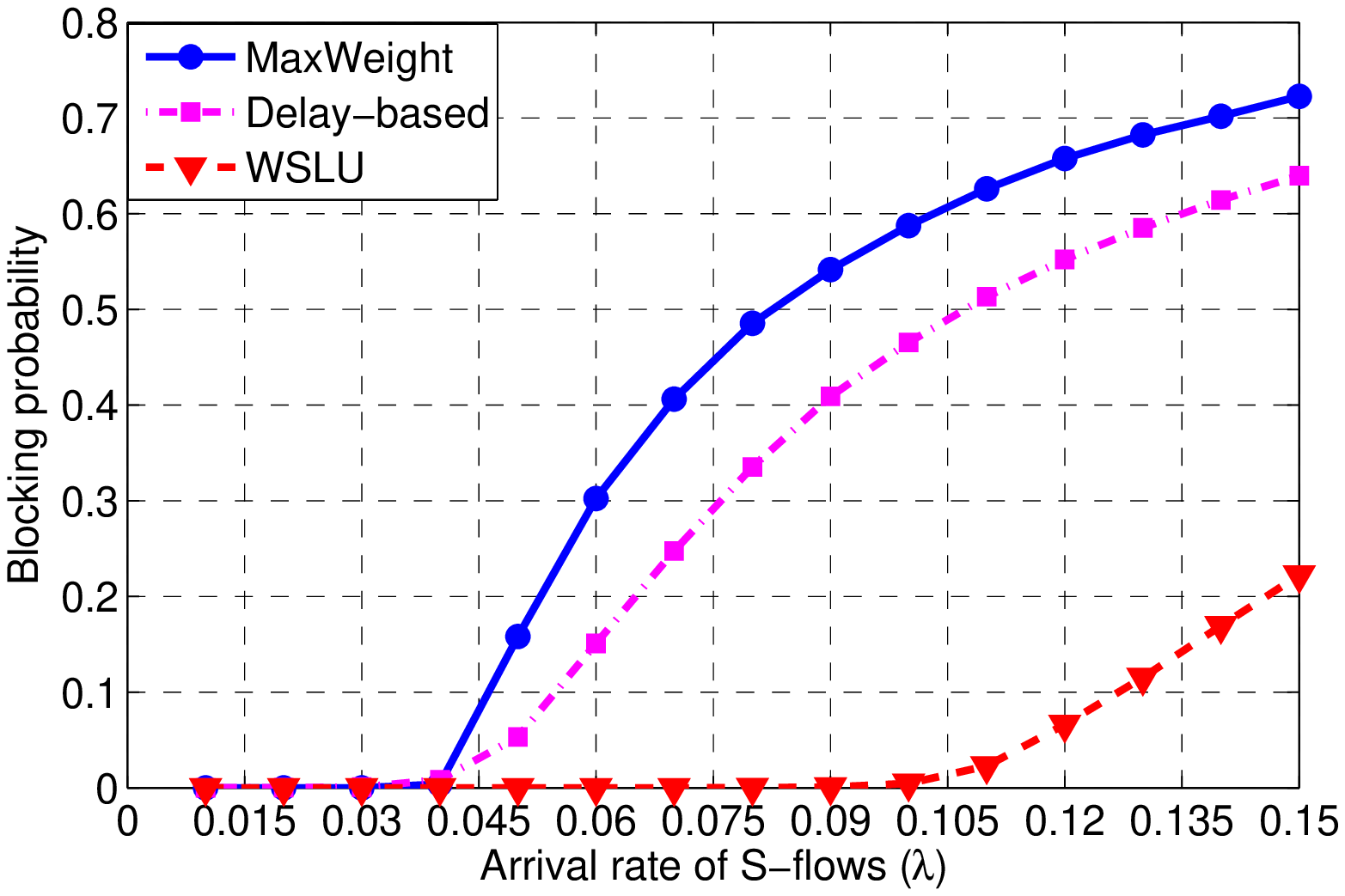}
\caption{The blocking probabilities of the Delay-based, MaxWeight, and WSLU  in a network with L-flows}
\label{fig: Blocking_2}
\end{center}
\end{figure}

\subsection*{Simulation V:  WSLU versus WSLO}

In this simulation, we study the impact of tie-breaking rules on performance. We compare the performance of the WSLU and WSLO. We first study the case where the base station does not limit the number of simultaneously active flows and there is no long-lived flow in the network. The simulation setting is the same as that in Simulation III. Figure \ref{fig: WSLUO_without} shows the average file-transfer delay and average number of S-flows under different values of $\lambda.$ We can see that the WSLO reduces the file-transfer delay and number of S-flows by nearly $75\%$ when $\lambda=0.13,$ which indicates the importance of
selecting a good tie-breaking rule for improving the network
performance.

\begin{figure}[h]
\begin{center}
\includegraphics[width=2.7in]{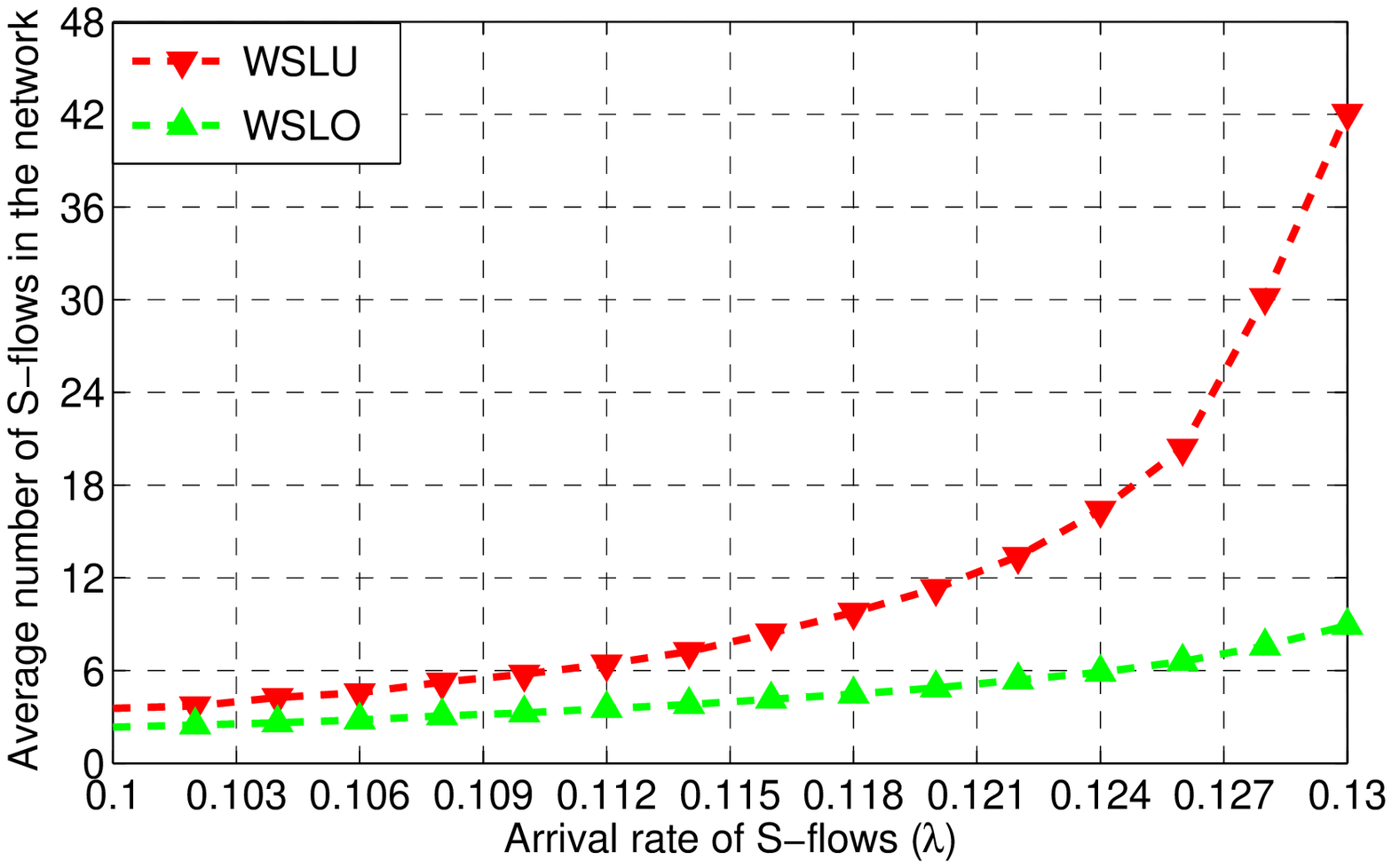}
\includegraphics[width=2.7in]{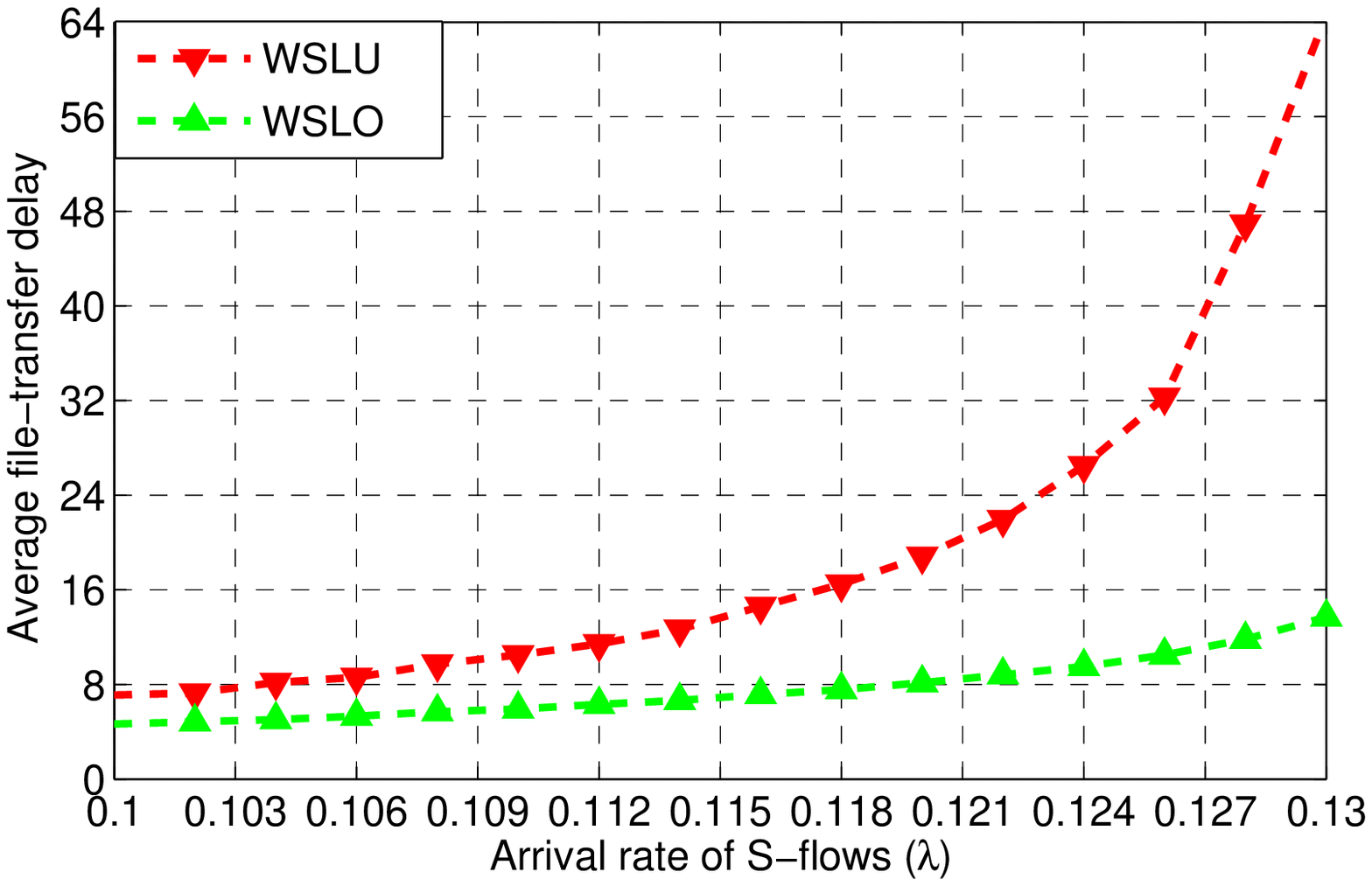}
\caption{The performance of the WSLU and WSLO algorithms in a network without L-flows}
\label{fig: WSLUO_without}
\end{center}
\end{figure}

Next, we study the case where the base station does not limit the number of simultaneously active flows and there are three L-flows in the network. Figure \ref{fig: WSLUO_with} shows the average number of short-lived flows and average file-transfer delay under different values of $\lambda.$ We can see again that the WSLO algorithm has a much better performance than the WSLU, especially when $\lambda$ is large.

\begin{figure}[h]
\begin{center}
\includegraphics[width=2.7in]{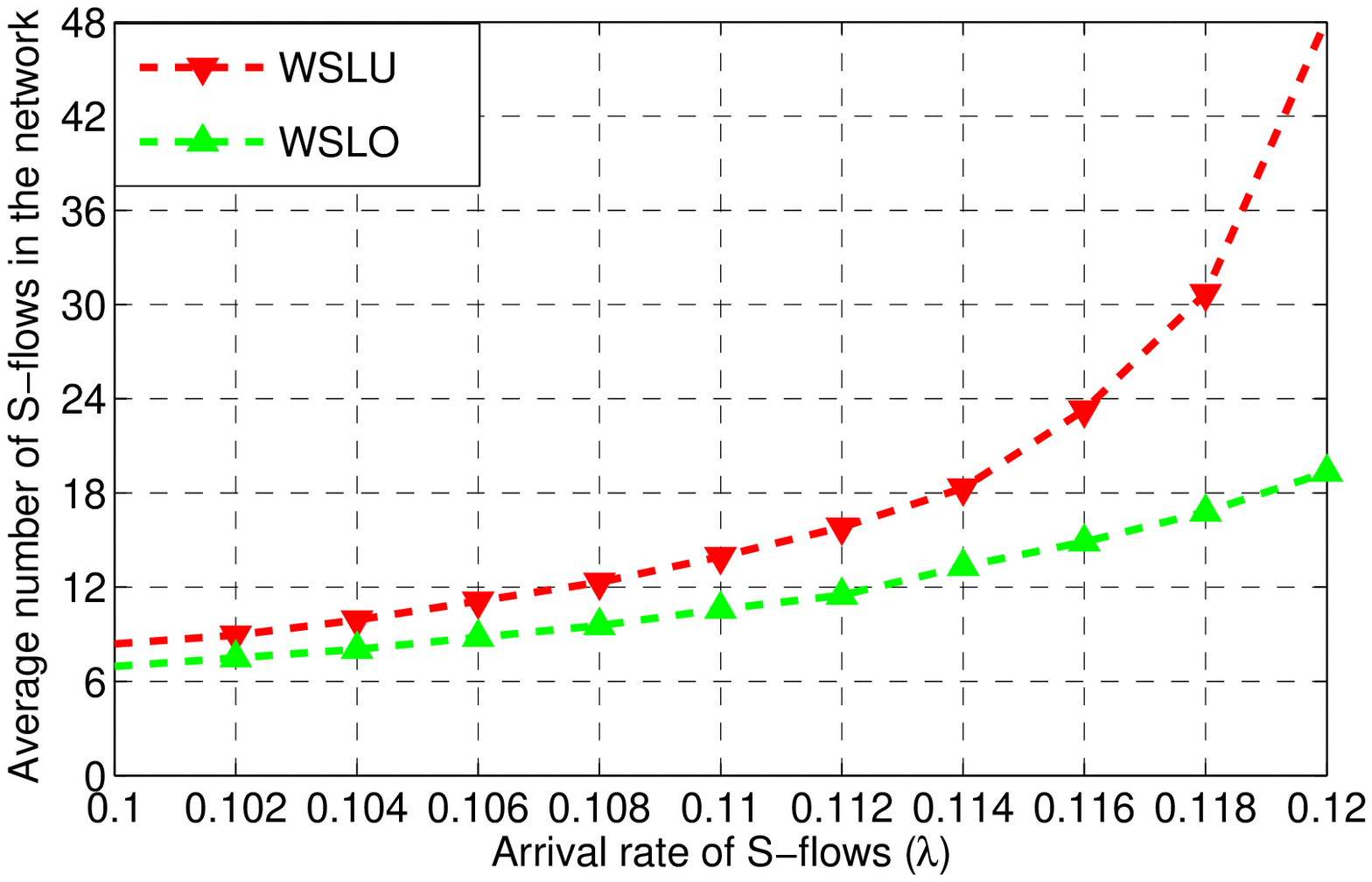}
\includegraphics[width=2.7in]{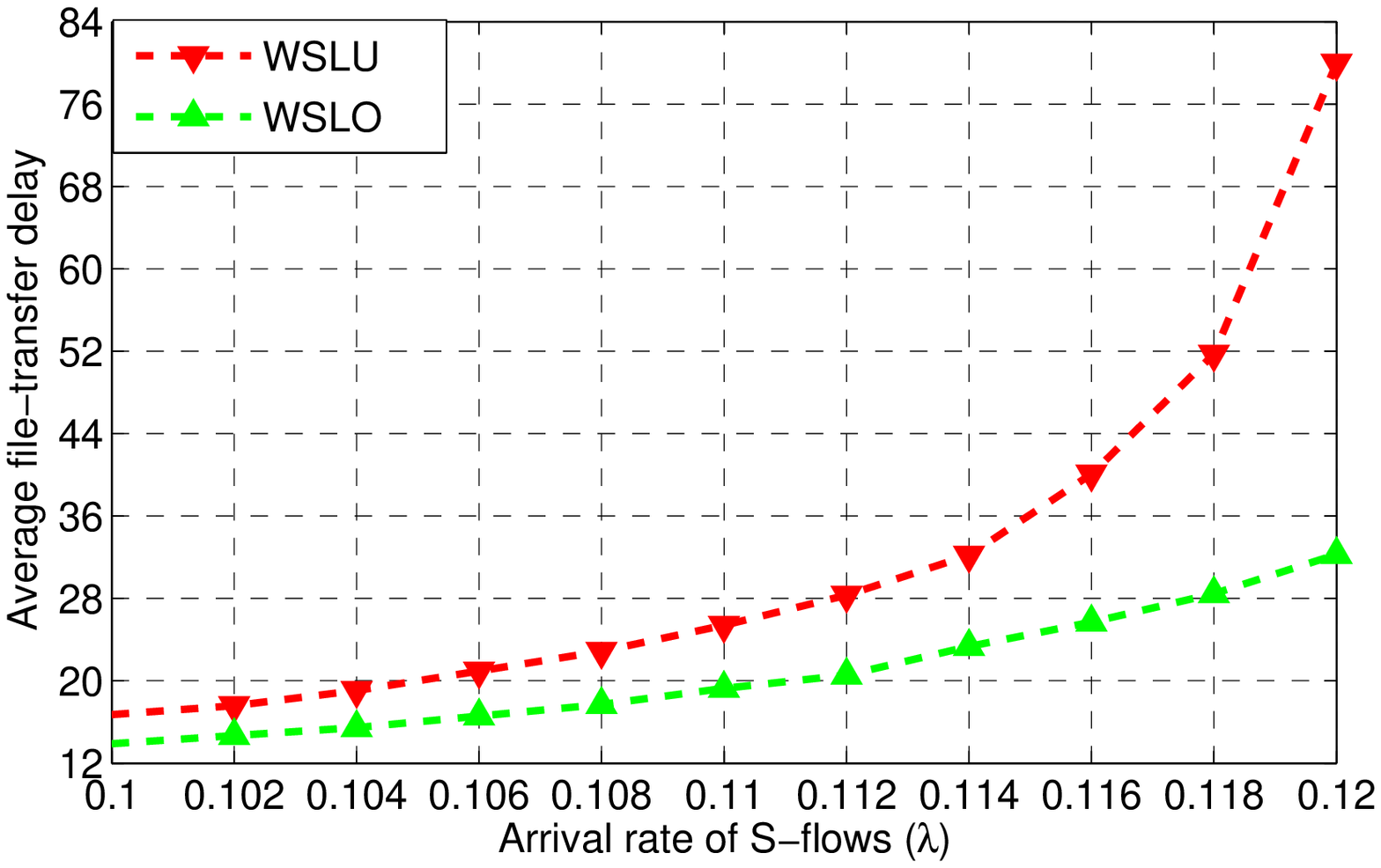}
\caption{The performance of the WSLU and WSLO algorithms in a network with both S-flows and L-flows}
\label{fig: WSLUO_with}
\end{center}
\end{figure}

Finally we consider the situation where the base station can support at most $20$ S-flows. A new S-flow will be blocked if $20$ S-flows are already in the network. The simulation setting is the same as that in Simulation IV. We calculate the blocking probabilities, and the results are shown in Figure \ref{fig: WSLUO_blocking_1} and \ref{fig: WSLUO_blocking_2}. We can see that the blocking probability under the WSLO is much smaller than that under the WSLU policy when $\lambda$ is large.
\begin{figure}[h]
\begin{center}
\includegraphics[width=2.7in]{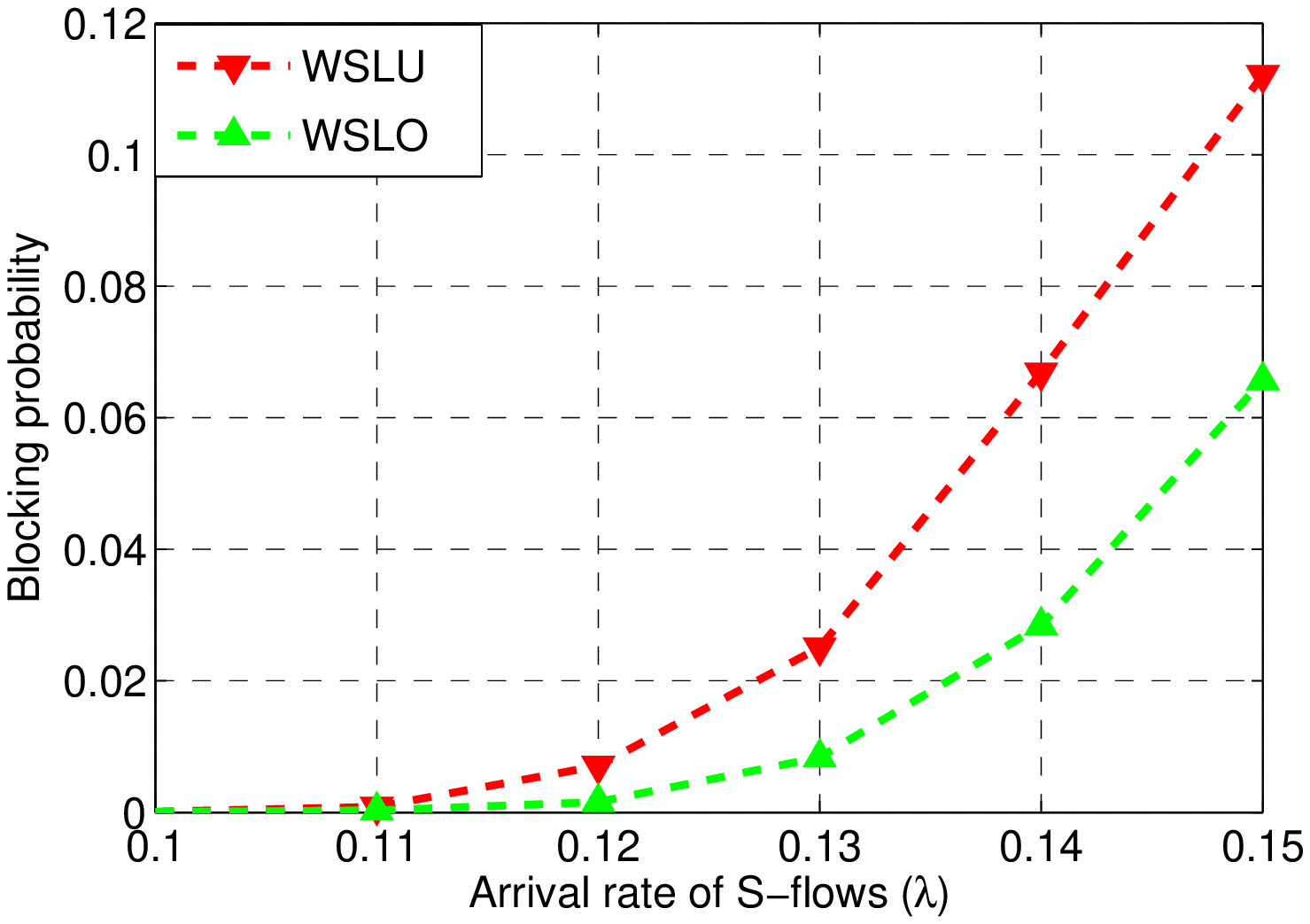}
\caption{The blocking probabilities of the WSLU and WSLO in a network without L-flows}
\label{fig: WSLUO_blocking_1}
\includegraphics[width=2.7in]{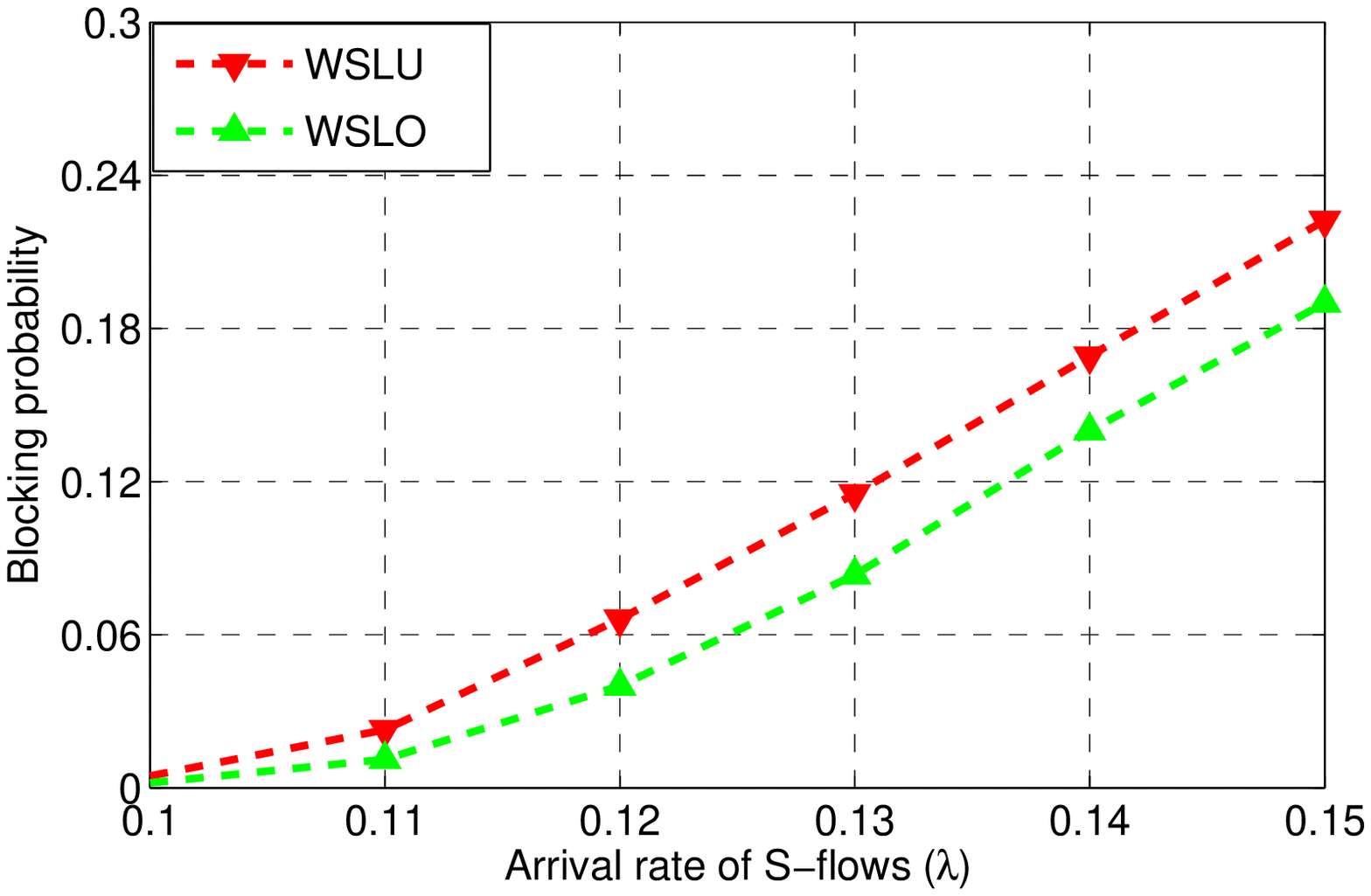}
\caption{The blocking probabilities of the WSLU and WSLO in a network with L-flows}
\label{fig: WSLUO_blocking_2}
\end{center}
\end{figure}

\section{Conclusions and Discussions}
In this paper, we studied multiuser scheduling in networks with flow-level dynamics. We first obtained necessary conditions for flow-level stability of networks with both long-lived flows and short-lived flows. Then based on an optimization framework, we proposed the workload-based scheduling with learning that is throughput-optimal under flow-level dynamics and requires no prior knowledge about channels and traffic. In the simulations, we evaluated the performance of the proposed scheduling algorithms, and demonstrated that the proposed algorithm performs significantly better than the MaxWeight algorithm and the Delay-based algorithm in various settings. Next we discuss the limitations of our model and possible extensions.

\subsection{The choice of $D$}
According to Theorem \ref{thm: learning},  the learning period $D$ should be sufficiently large to guarantee throughput-optimality. Our simulation results on the other hand suggested that a small $D$ may result in better performance. Therefore, there is clear trade-off in choosing $D.$ The study of the choice for $D$ is one potential future work.

\subsection{Unbounded file arrivals and file sizes }
One limitation of our model is that the random variables associated with the number of file arrivals and file sizes are assumed to be upper bounded. One interesting future research problem is to extend the results to unbounded number of file arrivals and file sizes.

%\section*{Acknowledgments}
%Research supported by NSF Grants 07-21286 and 08-31756, ARO MURI Subcontracts, and the DTRA grants HDTRA1-08-1-0016 and HDTRA1-09-1-0055.

\section*{Appendix A: Proof of Theorem \ref{thm: ori}}
Recall that $W_s(t)=\sum_{i\in {\cal I}(t)} \left\lceil\frac{Q_i(t)}{{R}_i^{\max}}\right\rceil.$ We define $\hat{R}_k^{\max}$ to be the largest achievable link rate of class-$k$ short-lived flows, and
$A_s(t)=\sum_{k\in{\cal K}}\sum_{i=1}^{{\Lambda}_k(t)} \left\lceil\frac{f_i}{\hat{R}^{\max}_k}\right\rceil,$
which is the amount of new workload (from short-lived flows) injected in the network at time $t,$ and $\mu_s(t)$ to be the decrease of the workload at time $t,$ i.e., $\mu_s(t)=1$ if the workload of short-lived flows is reduced by one and $\mu_s(t)=0$ otherwise. Based on the notations above, the evolution of short-lived flows can be described as:  $$W_s(t+1)=W_s(t)+A_s(t)-\mu_s(t).$$ Further, the evolution of  $Q_l(t)$ can be described as $$Q_l(t+1)=Q_l(t)+X_l(t)-\mu_l(t)+u_l(t),$$ where $\mu_l(t)$ is the decrease of $Q_l(t)$ due to the service long-lived flow $l$ receives at time $t,$ and $u_l(t)$ is the unused service due to the lack of data in the queue.

We consider the following Lyapunov function
\begin{equation}
V(t)=\alpha\left(W_s(t)\right)^2+\sum_{l\in{\cal L}}(Q_l(t))^2.
\end{equation}
We will prove that the drift of the Lyapunov function satisfies
\begin{eqnarray*}
\Ex[V(t+1)-V(t)|{\bf M}(t)] \leq U_d1_{{\bf M}(t)\in\Upsilon} -\frac{\epsilon}{2}\left[\alpha \bar{\lambda}W_s(t) \right.\\
+\left. \sum_{l\in \cal{L}} Q_l(t)x_l \right] 1_{{\bf M}(t)\not\in\Upsilon}
\end{eqnarray*} for some $U_d>0,$ $\bar{\lambda}>0$ and a finite set $\Upsilon$  (the values of these parameters will be defined in the following analysis). Positive recurrence of $\mathbf{M}$ then follows from Foster's Criterion for Markov chains \cite{asm03}.

First, since the number of arrivals, the sizes of short-lived flows and channel rates are all bounded,
it can be verified that there exists $U,$ independent of $\mathbf{M}(t),$ such that
\begin{align*}
&\Ex[V(t+1)-V(t)|{\bf M}(t)]\\
=&\Ex\left[\alpha\left(W_s(t+1)\right)^2-\alpha\left(W_s(t)\right)^2+\right.\\
&\hspace{0.2in}\left.\left.\sum_{l\in{\cal L}}(Q_l(t+1))^2-\sum_{l\in{\cal L}}(Q_l(t))^2\right|\mathbf{M}(t)\right]\\
\leq&U+2\alpha W_s(t) \Ex\left[\left.A_s(t)-\mu_s(t)\right|\mathbf{M}(t)\right]+\\
&\hspace{0.4in}2\sum_{l\in{\cal L}}Q_l(t)\Ex\left[\left.X_l(t)-\mu_l(t)\right|\mathbf{M}(t)\right]\\
\leq&U+2\alpha W_s(t)\left(\left(\sum_{k\in{\cal K}}\lambda_k \Ex\left[\left\lceil\frac{\hat{F}_k}{\hat{R}^{\max}_k}\right\rceil\right]\right)\right.\\
&\hspace{1.6in}-\Ex\left[\left.\mu_s(t)\right|\mathbf{M}(t)\right]\Big)\\
&\hspace{0.2in}+2\sum_{l\in{\cal L}}Q_l(t)\left(x_l-\Ex\left[\left.\mu_l(t)\right|\mathbf{M}(t)\right]\right).
\end{align*}
Recall that we assume that $(1+\epsilon)x_l$ and $(1+\epsilon)\lambda_k$ satisfy the supportability conditions of Theorem \ref{thm: NC}. By adding and subtracting corresponding $p_{\mathbf{c},l}R_{\mathbf{c},l}$ and $\mu_{\mathbf{c},s},$ we obtain that
\begin{eqnarray*}
&&\Ex[V(t+1)-V(t)|{\bf M}(t)]-U\\
&\leq&2\alpha W_s(t) \Ex\left[\left.\Ex\left[\left.\mu_{\mathbf{c},s}-\mu_{s}(t)\right|\mathbf{C}(t)=\mathbf{c}\right]\right|\mathbf{M}(t)\right]\\
&&+2\sum_{l\in{\cal L}}Q_l(t)\Ex\left[\left.\Ex\left[\left.p_{\mathbf{c},l}R_{\mathbf{c},l}-\mu_l(t)\right|\mathbf{C}(t)=\mathbf{c}\right]\right|\mathbf{M}(t)\right]\\
&&-2\epsilon\alpha W_s(t) \bar{\lambda}-2\epsilon\sum_{l\in{\cal L}}Q_l(t)x_l,
\end{eqnarray*} where
$$\bar{\lambda}=\left(\sum_{k\in{\cal K}}\lambda_k \Ex\left[\left\lceil\frac{\hat{F}_k}{\hat{R}^{\max}_k}\right\rceil\right]\right).$$

Next we assume $\mathbf{C}(t)=\mathbf{c}$ and analyze the following quantity
\begin{align}
\displaystyle \alpha W_s(t)\left(\mu_{\mathbf{c},s}-\mu_s(t)\right)+\sum_{l\in{\cal L}}Q_l(t)\left(p_{\mathbf{c},l}R_{\mathbf{c},l}-\mu_l(t)\right).\label{eq: dif}
\end{align} We have the following facts:
\begin{itemize}
\item {\bf Fact 1:} {\em Assume that there exists a short-lived flow $i$ such that $R_i(t)=R^{\max}_i$ or $R_i(t)\geq Q_i(t).$} If a short-lived flow is selected to be served, then the workload of the selected flow is reduced by one and $\mu_s(t)=1.$  If long-lived flow $l$ is selected, the rate flow $l$ receives is $R_{\mathbf{c},l}.$ Thus, we have that
\begin{eqnarray*}
&\displaystyle \alpha  W_s(t) \mu_s(t)+\sum_{l\in{\cal L}}Q_l(t)\mu_l(t) \\
=&\max\left\{\alpha  W_s(t), \max_l Q_l(t)R_{\mathbf{c},l}\right\}\\
\geq&\displaystyle  \alpha  W_s(t) \mu_{\mathbf{c},s}+\sum_{l\in{\cal L}}Q_l(t)p_{\mathbf{c},l}R_{\mathbf{c},l},
\end{eqnarray*} where the last inequality holds because $\sum_l p_{\mathbf{c},l}+\mu_{\mathbf{c},s}\leq 1.$
Therefore, we have $(\ref{eq: dif})\leq 0$ in this case.

\item {\bf Fact 2:} {\em Assume that there does not exist a short-lived flow $i$ such that $R_i(t)=R^{\max}_i$ or $R_i(t)\geq Q_i(t).$} In this case, we have
\begin{eqnarray*}
(\ref{eq: dif})&\leq& \alpha W_s(t)+\max_{l\in{\cal L}}Q_l(t)R_{\mathbf{c},l}\\
&\leq& \alpha W_s(t)+R^{\max}\max_{l\in{\cal L}}Q_l(t).
\end{eqnarray*}
\rightline{$\square$}
\end{itemize}

Now we define a set $\Upsilon$ such that $$\Upsilon=\left\{\mathbf{M}:  W_s\leq {U_W} \hbox{ and } Q_l\leq U_Q \hbox{ }\forall l\right\},$$ where $U_W$ is a positive integer satisfying that
\begin{eqnarray}
&(1-p^{\max}_s)^{\frac{U_W}{{F}^{\max}}}\leq \frac{{\epsilon}}{2}\min\left\{\bar{\lambda}, \frac{\min_{l\in{\cal L}} x_l}{R^{\max}}\right\}\triangleq\epsilon_1\label{eq: UW1}\\
&U_W\geq \frac{2U}{\epsilon \alpha \bar{\lambda}},\label{eq: UW2}\end{eqnarray} and $U_Q$ is a positive integer satisfying \begin{eqnarray}
U_Q\geq \frac{4\alpha U_W+U}{\epsilon\min_{l\in{\cal L}}x_l}.\label{eq: UQ0}
\end{eqnarray}

We next compute the drift of the Lyapunov function according to the value of $\mathbf{M}(t).$
\begin{itemize}
\item {\bf Case I:} Assume $\mathbf{M}(t)\in \Upsilon.$ According to the definition of $\Upsilon,$ we have \begin{eqnarray*}
\Ex[V(t+1)-V(t)|{\bf M}(t)]\leq U+2\alpha U_W+ 2{R}^{\max} L U_Q.
\end{eqnarray*}

\item  {\bf Case II:} Assume  $W_s(t)> U_W.$ Since the size of a short-lived flow is upper bounded by $F^{\max},$ $ W_s(t)> U_W$ implies that at least $\frac{U_W}{ {F}^{\max}}$ short-lived flows are in the network at time $t.$ Define ${\cal S}(t)$ to be the following event: \begin{color}{black}no short-lived flow satisfies $R_i(t)=R_i^{\max}$ or $R_i(t)\geq Q_i(t)$ . \end{color} 

    Recall that  $$\min_i \Pr(R_i(t)=R^{\max}_i)\geq p_s^{\max}.$$ Given at least $\frac{U_W}{{F}^{\max}}$ short-lived flows are in the network, we have that $$\Pr(1_{{\cal S}(t)}=1)\leq (1-p^{\max}_s)^{\frac{U_W}{{F}^{\max}}}\leq \epsilon_1.$$

    According to facts 1 and 2, $(\ref{eq: dif})$ is positive only if ${\cal S}(t)$ occurs and the value of $(\ref{eq: dif})$ is bounded by $\alpha  W_s(t)+R^{\max}\max_{l\in{\cal L}}Q_l(t).$ Therefore, we can conclude that in this case (Case II),
    \begin{eqnarray}
&&\Ex[V(t+1)-V(t)|{\bf M}(t)]\nonumber\\
&\leq& U+2 \epsilon_1 \left(\alpha  W_s(t)+R^{\max}\max_{l\in{\cal L}}Q_l(t)\right)\nonumber\\
&&-2\epsilon \alpha  W_s(t) \bar{\lambda}-2\epsilon\sum_{l\in{\cal L}}Q_l(t)x_l \nonumber\\
&\leq& U-\epsilon \alpha W_s(t) \bar{\lambda}-\epsilon\sum_{l\in{\cal L}}Q_l(t)x_l \label{eq: II1}\\
&\leq& -\frac{\epsilon}{2}\left[\alpha \bar{\lambda}W_s(t)+ \sum_{l\in \cal{L}} Q_l(t)x_l \right]\label{eq: II2}
\end{eqnarray} where inequality (\ref{eq: II1}) holds due to the definition of $\epsilon_1$ (\ref{eq: UW1}), and inequality (\ref{eq: II2}) holds due to inequality (\ref{eq: UW2}).

\item {\bf Case III:} Assume that $ W_s(t)\leq U_W$ and $Q_l(t)> U_Q$ for some $l.$ In this case,
if a long-lived flow is selected for a given $\mathbf{c},$ we have
\begin{eqnarray*}
(\ref{eq: dif})&\leq&\alpha W_s(t)\mu_{\mathbf{c},s} \leq  \alpha W_s(t).
\end{eqnarray*}  Otherwise, if a short-lived flow is selected, it means for the given $\mathbf{c},$ we have
$\max_l Q_l(t)R_{\mathbf{c},l}\leq \alpha  W_s(t),$
and
\begin{eqnarray*}
(\ref{eq: dif})\leq 2\alpha W_s(t).
\end{eqnarray*}
Therefore, we can conclude that in this case,
\begin{align}
&\Ex[V(t+1)-V(t)|{\bf M}(t)]\nonumber\\
\leq& U+4\alpha W_s(t)-2\epsilon\alpha W_s(t) \bar{\lambda}-2\epsilon\sum_{l\in{\cal L}}Q_l(t)x_l\label{eq: use1}\\
\leq& U+4\alpha U_W-2\epsilon\alpha W_s(t) \bar{\lambda}-2\epsilon\sum_{l\in{\cal L}}Q_l(t)x_l\nonumber\\
\leq&-\frac{\epsilon}{2}\left[\alpha \bar{\lambda}W_s(t)+ \sum_{l\in \cal{L}} Q_l(t)x_l \right]
\end{align} where the last inequality yields from the definition of $U_Q$ (\ref{eq: UQ0}).
\end{itemize}

From the analysis above, we can conclude that
\begin{eqnarray*}
\Ex[V(t+1)-V(t)|{\bf M}(t)] \leq U_d1_{{\bf M}(t)\in\Upsilon} -\frac{\epsilon}{2}\left[\alpha \bar{\lambda}W_s(t) \right.\\
+\left. \sum_{l\in \cal{L}} Q_l(t)x_l \right] 1_{{\bf M}(t)\not\in\Upsilon},
\end{eqnarray*} where $U_d= U+2\alpha U_W+ 2{R}^{\max} L U_Q$ and $\Upsilon$ is a set with a finite number of elements. \begin{color}{black} Since $V(t)\geq 0$ for all $t,$ the Lyapunov function is always lower bounded. Further the drift of the Lyapunov is upper bounded when ${\bf M}(t)$ belongs to a finite set $\Upsilon,$ and is negative otherwise. \end{color} So invoking Foster's criterion, the Markov chain $\mathbf{M}(t)$ is positive recurrent and the boundedness of the first moment follows from \cite{MeyTwe_09}.

\section*{Appendix B: Proof of Theorem \ref{thm: learning}}
Consider the network that is operated under WSL, and define ${\cal H}(t)$ to be $${\cal H}(t)\triangleq\left\{Q_l(t), R_l(t), Q_i(t), R_i(t),\tilde{R}^{\max}_i(t)\right\}.$$ Now {\em given ${\cal H}(t),$} we define the following notations:
\begin{itemize}
\item Define $\mu_{2;l}(t)=R_l(t)$ if flow $l$ is selected by WSL, and $\mu_{2;l}(t)=0$ otherwise.

\item Define $\mu_{2;i}(t)=1$ if flow $i$ is selected by WSL and the workload of flow $i$ can be reduced by one, and $\mu_{2;i}(t)=0$ otherwise.

\item Define $\mu_{1;l}(t)=R_l(t)$ if flow $l$ is selected by WS, and $\mu_{1;l}(t)=0$ otherwise.

\item Define $\mu_{1;i}(t)=1$ if flow $i$ is selected by WS and the workload of flow $i$ can be reduced by one, and $\mu_{1;i}(t)=0$ otherwise.
\end{itemize}
We remark that $\mu_{2;j}(t)$ is the action selected by the base station at time $t$ under WSL and $\mu_{1;j}(t)$ is the action selected by the base station at time $t$ under WS, assuming the same history ${\cal H}(t).$

We define the Lyapunov function to be \begin{equation}
V(n)=\alpha \left(W_s(nT)\right)^2+\sum_{l\in{\cal L}}(Q_l(nT))^2.
\end{equation}
\begin{color}{black} This Lyapunov function is similar to the one used in the proof of Theorem \ref{thm: ori}, and we will show that this is a valid Lyapunov function for the workload-based scheduling with learning. \end{color} Then, it is easy to verify that there exists $U_1$ independent of $\tilde{\bf M}(n)$ such that
\begin{align*}
&\Ex[V(n+1)-V(n)|\tilde{\mathbf{M}}(n)]\nonumber\\
<&U_1+2\alpha \Ex\left[W_s(nT)\left.\sum_{t=nT}^{(n+1)T-1}\left(A_s(t)-\mu_{2;s}(t)\right)\right|\tilde{\mathbf{M}}(n)\right]\nonumber\\
&+\sum_{l\in{\cal L}}2\Ex\left[\left.Q_l(nT)\sum_{t=nT}^{(n+1)T-1}\left(X_l(t)-\mu_{2;l}(t)\right)\right|\tilde{\mathbf{M}}(n)\right].
\end{align*}
Dividing the time into two segments $[nT, nT+D-1]$ and $[nT+D, (n+1)T-1],$ we obtain
\begin{align*}
&\Ex[V(n+1)-V(n)|\tilde{\mathbf{M}}(n)]\nonumber\\
<&U_1+2\alpha W_s(nT) \bar{\lambda} D+2\sum_{l\in{\cal L}} Q_l(nT)x_lD\\
&+2\alpha \Ex\left[W_s(nT)\left.\sum_{t=nT+D}^{(n+1)T-1}\left(A_s(t)-\mu_{2;s}(t)\right)\right|\tilde{\mathbf{M}}(n)\right]\nonumber\\
&+\sum_{l\in{\cal L}}2\Ex\left[\left.Q_l(nT)\sum_{t=nT+D}^{(n+1)T-1}\left(X_l(t)-\mu_{2;l}(t)\right)\right|\tilde{\mathbf{M}}(n)\right].
\end{align*}

Note that $|Q_l(t_1)-Q_l(t_2)|$ and $|W_k(t_1)-W_k(t_2)|$ are both bounded by some constants independent of $\tilde{\mathbf{M}}(n),$ so there exists $\tilde{U}$ such that
\begin{align*}
&\Ex[V(n+1)-V(n)|\tilde{\mathbf{M}}(n)]\nonumber\\
<&\tilde{U}+2\alpha W_s(nT) \bar{\lambda} D+2\sum_{l\in{\cal L}} Q_l(nT)x_lD\\
&+2\Ex\left[\left.\alpha\sum_{t=nT+D}^{(n+1)T-1}W_s(t)\left(A_s(t)-\mu_{2;s}(t)\right)\right|\tilde{\mathbf{M}}(n)\right]\nonumber\\
&+\sum_{l\in{\cal L}}2\Ex\left[\left.\sum_{t=nT+D}^{(n+1)T-1}Q_l(t)\left(X_l(t)-\mu_{2;l}(t)\right)\right|\tilde{\mathbf{M}}(n)\right].
\end{align*}
Now, by adding and subtracting $\mu_{1;\cdot}(t),$ we obtain
\begin{align*}
&\Ex[V(n+1)-V(n)|\tilde{\mathbf{M}}(n)]\\
\leq&\tilde{U}+2\alpha W_s(nT) \bar{\lambda} D+2\sum_{l\in{\cal L}} Q_l(nT)x_lD+\sum_{t=nT+D}^{(n+1)T-1} \hbox{Drift}(t),
\end{align*} where
\begin{align}
&\hbox{Drift}(t)\nonumber\\
=&2\Ex\left[\left.\alpha W_s(t)A_s(t)+\sum_{l\in{\cal L}}Q_l(t)X_l(t)\right|\tilde{\mathbf{M}}(n)\right]\label{eq: 1:1}\\
&-2\Ex[\alpha W_s(t)\mu_{1;s}(t)+\sum_{l\in{\cal L}} Q_l(t)\mu_{1;l}(t)|\tilde{\mathbf{M}}(n)]\label{eq: 1:2}\\
&+\sum_{l\in{\cal L}}2\Ex[Q_l(t)\left(\mu_{1;l}(t)-\mu_{2;l}(t)\right)|\tilde{\mathbf{M}}(n)]\label{eq: 2:2}\\
&+2\Ex\left[\alpha W_s(t)\left(\mu_{1;s}(t)-\mu_{2;s}(t)\right)|\tilde{\mathbf{M}}(n)\right].\label{eq: 4}
\end{align}

Note that (\ref{eq: 2:2})+(\ref{eq: 4}) is the difference between WS and WSL. In the following analysis, we will prove that this difference is small compared to the absolute value of (\ref{eq: 1:1})+(\ref{eq: 1:2}).

We define \begin{align*}
{\hbox{Diff}(t)}=&\alpha W_s(t)\left(\mu_{1;s}(t)-\mu_{2;s}(t)\right)\\
&+\sum_{l\in{\cal L}}Q_l(t)\left(\mu_{1;l}(t)-\mu_{2;l}(t)\right),
\end{align*}
and
$$\tilde{W}_s(t)=\sum_{i\in {\cal I}(t)} \left\lceil\frac{Q_i(t)}{\tilde{R}_i^{\max}(t)}\right\rceil.$$
Next, we compute its value in three different situations:
\begin{itemize}
\item {\bf Situ-A:}  {\em Consider the situation in which $\alpha \tilde{W}_s(t)\leq \max_{l\in{\cal L}}Q_l(t)R_l(t).$} We note that $\tilde{W}_s(t)\geq W_s(t)$ since $\tilde{R}_i^{\max}(t)\leq R_i^{\max}$ for all $t$ and $i.$ Therefore, given $\alpha \tilde{W}_s(t)\leq \sum_{l\in{\cal L}}Q_l(t),$  both WS and WSL will select a long-lived flow. In this case, we can conclude that $$\mu_{1;l}(t)=\mu_{2;l}(t)\hbox{ and }\mu_{1;s}(t)=\mu_{2;s}(t)=0,$$ and
$\hbox{Diff}(t)=0.$

\item {\bf Situ-B:} {\em Consider the situation in which $\alpha {W}_s(t)> \max_{l\in{\cal L}}Q_l(t)R_l(t).$} In this case, both WS and WSL will select a short-lived flow, which implies that
$$\mu_{1;l}(t)=\mu_{2;l}(t)=0,$$ and
\begin{align*}
{\hbox{Diff}(t)}=&\alpha W_s(t)\left(\mu_{1;s}(t)-\mu_{2;s}(t)\right)\\
&\leq \alpha W_s(t)\left(1-\mu_{2;s}(t)\right).
\end{align*}

\item {\bf Situ-C:} {\em Consider the situation in which $\alpha \tilde{W}_s(t)> \max_{l\in{\cal L}}Q_l(t)R_l(t)\geq \alpha W_s(t).$} In this case, WS will select a long-lived flow and WSL will select a short-lived flow. We hence have
$$\mu_{1;l}(t)>0 \hbox{ and } \mu_{1;s}(t)=\mu_{2;l}(t)=0,$$ and
\begin{align*}
{\hbox{Diff}(t)}&=\max_{l\in{\cal L}}Q_l(t)R_l(t)-\alpha W_s(t) \mu_{2;s}(t)\\
&\leq \alpha \tilde{W}_s(t)-\alpha W_s(t) \mu_{2;s}(t)
\end{align*}
\end{itemize}
\rightline{$\square$}

According to the analysis above, we have that
\begin{align*}
&\Ex[{\hbox{Diff}(t)}|\tilde{\mathbf{M}}(n)]\\
\leq&  \Ex\left[\alpha W_s(t)|\hbox{Situ-B}, \mu_{2;s}=0, \tilde{\mathbf{M}}(n)\right]\times\\
&\Pr\left(\hbox{Situ-B}, \mu_{2;s}=0|\tilde{\mathbf{M}}(n)\right)\\
+& \Ex\left[\alpha\tilde{W}_s(t)|\hbox{Situ-C}, \mu_{2;s}=0, \tilde{\mathbf{M}}(n)\right]\times\\
&\Pr\left(\hbox{Situ-C}, \mu_{2;s}=0|\tilde{\mathbf{M}}(n)\right)\\
+& \Ex\left[\alpha\tilde{W}_s(t)-\alpha W_s(t)|\hbox{Situ-C}, \mu_{2;s}=1, \tilde{\mathbf{M}}(n)\right]\times\\
&\Pr\left(\hbox{Situ-C}, \mu_{2;s}=1|\tilde{\mathbf{M}}(n)\right).
\end{align*}

Next we define a finite set $\tilde{\Upsilon}.$ We first introduce some constants:
\begin{itemize}
\item $\epsilon_1=\min\left\{\frac{\bar{\lambda}\epsilon}{32}, \frac{\epsilon\min_l x_l}{8R^{\max}}\right\}.$

\item $\epsilon_2=\frac{\bar{\lambda}\epsilon}{32 R^{\max} },$ and $D_{\epsilon_2}$ and $N_{\epsilon_2}$ are the numbers that guarantee $\Pr\left({\cal E}_{\scriptsize miss}(t)\right)\leq \epsilon_2,$ which are defined by the goodness of the tie-breaking rule.

\item $\lambda^{\max}_W=K\lambda^{\max}F^{\max},$ which is the maximum number of bits of short-lived flows injected in one time slot, and also the upper bound on the new workload injected in the network in one time slot.
\end{itemize}
We define a set $\tilde{\Upsilon}$ such that $$\tilde{\Upsilon}=\left\{\tilde{\mathbf{M}}(n):   \substack{W_s(nT)\leq \tilde{U}_W +2T+\frac{2\sum_l x_l R^{\max}T}{\alpha\bar{\lambda}}\\ Q_l(nT)\leq \tilde{U}_Q +\frac{2\alpha\bar{\lambda}T}{\min_l x_l}+\frac{2T R^{\max} \sum_l x_l}{\min_l x_l}\hbox{ }\forall l}\right\}.$$ In this definition, $\tilde{U}_W$ is a positive integer satisfying that
\begin{eqnarray}
&(1-p^{\max}_s)^{\frac{\tilde{U}_W}{{F}^{\max}}}\leq  \epsilon_1,\label{eq: UW12}\\
&\tilde{U}_W\geq \frac{\frac{8\tilde{U}}{T-D}+16\epsilon_2\alpha\lambda^{\max}_W T+8\alpha DR^{\max}+16\epsilon_2 \alpha R^{\max}T+8 \lambda^{\max}_W D}{\epsilon \alpha \bar{\lambda}}\label{eq: UW22}\\
&\frac{\tilde{U}_W}{F^{\max}} \geq N_{\epsilon_2} \label{eq: UW23},
\end{eqnarray} and $\tilde{U}_Q$ is a positive integer satisfying \begin{eqnarray}
\displaystyle \tilde{U}_Q\geq \textstyle {\frac{8\tilde{U}+12\alpha R^{\max}(\tilde{U}_W+\frac{2\sum_l x_l R^{\max}T}{\alpha\bar{\lambda}}+(\lambda^{\max}_W+2) T)}{\epsilon \min_l x_l}.}\label{eq: UQ}
\end{eqnarray}
Since the changes of $W_s(t)$ and $Q_l(t)$ during each time slot is bounded by some constants independent of $\tilde{\mathbf{M}}(n),$ it is easy to verify that $\tilde{\Upsilon}$ is a set of a finite number of elements.

Next, we analyze the drift of Lyapunov function case by case assuming that
\begin{eqnarray}
D>\left\lceil \frac{\log{\bar{\lambda}\epsilon}-\log{16}-\log{R^{\max}}}{\log(1-p_s^{\max})}\right\rceil \label{eq: D_condition}
\end{eqnarray}
and $T>\left\lceil \frac{(4+\epsilon)D}{\epsilon}\right\rceil.$

\begin{list}{\labelitemi}{\leftmargin=1em}

\item {\bf Case I:} Assume that $\tilde{M}(n)\in\tilde{\Upsilon}.$ In this case, it is easy to verify that $\Ex[V(n+1)-V(n)|\tilde{\mathbf{M}}(n)]$ is bounded by some constant $\tilde{U}_d.$

\item {\bf Case II:} Assume that $$W_s(nT)> \tilde{U}_W +2T+\frac{2\sum_l x_l R^{\max}T}{\alpha\bar{\lambda}} \geq \tilde{U}_W+T.$$ Recall that ${\cal E}_{\scriptsize miss}(t)$ is the event such that the tie-breaking rule selects a short-lived flow with $\tilde{R}^{\max}_i(t)\not=R^{\max}_i.$  Note that $\mu_{2;s}(t)=0$ implies that ${\cal E}_{\scriptsize miss}(t)$ occurs. Also note the following facts:
\begin{itemize}
\item[-] For any $nT\leq t\leq (n+1)T,$ we have $W(t)\leq W(nT)+\lambda^{\max}_W T,$

\item[-] Given $W_s(nT)\geq \tilde{U}_W+T,$ we have $ W_s(t)\geq \tilde{U}_W$ for all $nT\leq t\leq (n+1)T-1.$ Then according to the definition of $\epsilon_2$ and $\tilde{U}_W$ and assumption that the tie-breaking rule is good, we have
$$\Pr\left({\cal E}_{\scriptsize miss}(t)\right)\leq \epsilon_2$$ for all $nT+D\leq t\leq (n+1)T-1.$

\item[-]  Given any $\tilde{\mathbf{M}}(n)$  and any $nT+D\leq t\leq (n+1)T-1,$ we have
\begin{align}
& \Ex\left[\alpha \tilde{W}_s(t)-\alpha W_s(t)|\hbox{Situ-C}, \mu_{2;s}=1, \tilde{\mathbf{M}}(n)\right]\times\nonumber\\
&\Pr\left(\hbox{Situ-C}, \mu_{2;s}=1|\tilde{\mathbf{M}}(n)\right)\nonumber\\
\leq&  \Ex\left[\alpha \tilde{W}_s(t)-\alpha W_s(t)|\tilde{\mathbf{M}}(n)\right]\nonumber\\
=& \Ex\left[\left.\Ex\left[\left.\alpha  \tilde{W}_s(t)- \alpha  W_s(t)\right|W_s(t-D)|\right]\right|\tilde{\mathbf{M}}(n)\right]\nonumber\\
\leq& \Ex\left[\alpha(1-p_s^{\max})^DW_s(t-D)R^{\max}+\alpha\lambda^{\max}_W D|\tilde{\mathbf{M}}(n)\right]\label{eq: 5}\\
\leq&\Ex\left[\alpha (1-p_s^{\max})^D(W_s(t)+D)R^{\max}+\alpha \lambda^{\max}_W D|\tilde{\mathbf{M}}(n)\right]\nonumber,
\end{align} where the inequality (\ref{eq: 5}) holds because at most $\lambda^{\max}_W D$ bits belonging to short-lived flows are in the network for less than $D$ time slots at time $t,$  and a flow having been in the network for at least $D$ time slots can estimate correctly its workload with a probability at least $1-(1-p_s^{\max})^D.$
\end{itemize}
Now according to the observations above, we can obtain that
\begin{align*}
&\Ex[{\hbox{Diff}(t)}|\tilde{\mathbf{M}}(n)]\\
\leq& \epsilon_2\alpha \left(W_s(nT)+\lambda_{W}^{\max}T\right)+ \epsilon_2\alpha \left(R^{\max}W_s(nT)+\lambda_{W}^{\max}T\right) \\
&+  \Ex\left[\alpha(1-p_s^{\max})^D(W_s(t)+D)R^{\max}+\alpha\lambda^{\max}_W D|\tilde{\mathbf{M}}(n)\right].
\end{align*}
Combining with the analysis leading to (\ref{eq: II1}) in Appendix A, we conclude that
\begin{align*}
&{\hbox{Drift}(t)}\\
\leq &2  \Ex\left[ \epsilon_1 \left(\alpha  W_s(t)+R^{\max}\max_{l\in{\cal L}}Q_l(t)\right)\right.\nonumber\\
&-\epsilon \alpha  W_s(t) \bar{\lambda}-\epsilon\sum_{l\in{\cal L}}Q_l(t)x_l\nonumber\\
&+\epsilon_2\alpha \left(W_s(nT)+\lambda_{W}^{\max}T\right)\\
&+ \epsilon_2\alpha \left(R^{\max}W_s(nT)+\lambda_{W}^{\max}T\right) \\
&+ \left.\alpha(1-p_s^{\max})^D(W_s(t)+D)R^{\max}+\alpha \lambda^{\max}_W D|\tilde{\mathbf{M}}(n)\right]\\
\leq& \Ex\left[\left.-\epsilon\left(\alpha \bar{\lambda} W_s(t)+\sum_{l\in {\cal L}}x_l Q_l(t)\right)\right|\tilde{\mathbf{M}}(n)\right],
\end{align*} where the last inequality holds due to (\ref{eq: UW22}).

\item {\bf Case III:} Assume that $$W_s(nT)< \tilde{U}_W +2T+\frac{2\sum_l x_l R^{\max}T}{\alpha\bar{\lambda}} $$ and $$Q_l(nT)>\tilde{U}_Q +\frac{2\alpha\bar{\lambda}T}{\min_l x_l}+\frac{2T R^{\max} \sum_l x_l}{\min_l x_l}>\tilde{U}_Q$$ for some $l.$ In this case, we have
\begin{align*}
{\hbox{Diff}(t)}\leq \alpha\tilde{W}_s(t)\leq \alpha R^{\max}W_s(t).
\end{align*} Combining with the analysis leading to (\ref{eq: use1}) in Appendix A, we have that
\begin{align*}
&{\hbox{Drift}(t)}\\
\leq&2\Ex\left[\alpha R^{\max}W_s(t)+2\alpha W_s(t)\right.\\
&\left.\left.-\epsilon\left(\alpha \bar{\lambda} W_s(t)+\sum_{l\in {\cal L}}x_l Q_l(t)\right)\right|\tilde{\mathbf{M}}(n)\right]\\
\leq& \Ex\left[\left.-\epsilon\left(\alpha \bar{\lambda} W_s(t)+\sum_{l\in {\cal L}}x_l Q_l(t)\right)\right|\tilde{\mathbf{M}}(n)\right],
\end{align*} where the last inequality holds due to (\ref{eq: UQ}).
\end{list}
\rightline{$\square$}

Now, combining case II and case III, we can obtain that
\begin{align*}
&\Ex[V(n+1)-V(n)|\tilde{\mathbf{M}}(n)]\\
\leq&\tilde{U}+2\alpha W_s(nT) \bar{\lambda} D+2\sum_{l\in{\cal L}} Q_l(nT)x_lD\\
&+\sum_{t=nT+D}^{(n+1)T-1} \Ex\left[\left.-\epsilon\left(\alpha \bar{\lambda} W_s(t)+\sum_{l\in {\cal L}}x_l Q_l(t)\right)\right|\tilde{\mathbf{M}}(n)\right]\\
\leq&\tilde{U}+2\alpha W_s(nT) \bar{\lambda} D+2\sum_{l\in{\cal L}} Q_l(nT)x_lD\\
&-\epsilon (T-D) \left(\alpha \bar{\lambda} W_s(nT)+\sum_{l\in {\cal L}}x_l Q_l(nT)\right)\\
&+\epsilon(T-D)(\alpha\bar{\lambda} T+\sum_{l\in{\cal L}} x_l R^{\max} T)\\
\leq&-\tilde{U}-\sum_{t=nT+D}^{(n+1)T-1} \Ex\left[\left.\frac{\epsilon}{2}\left(\alpha \bar{\lambda} W_s(t)+\sum_{l\in {\cal L}}x_l Q_l(t)\right)\right|\tilde{\mathbf{M}}(n)\right],
\end{align*} where the last inequality yields from the definition of $\tilde{U}_W$ and $\tilde{U}_Q.$ Finally, we can conclude the theorem from \cite{asm03,MeyTwe_09}.

\section*{Appendix C: The Uniform Tie-breaking Rule}

Recall that we define ${\cal E}_{\scriptsize miss}(t)$ to be the event that the tie-breaking rule selects a short-lived flow with $\tilde{R}^{\max}_i(t)\not=R^{\max}_i.$

\begin{prop}
The uniform tie-breaking rule is good.
\end{prop}
\begin{proof}
Suppose set $${\cal J}(t)=\left\{i: R_i(t)=\tilde{R}_i^{\max}(t)\hbox{ or } R_i(t)\geq Q_i(t)\right\}.$$ Under the uniform tie-breaking, ${\cal E}_{\scriptsize miss}(t)$ occurs with probability
\begin{eqnarray*}
&&\frac{\left|\left\{i: i\in {\cal J}(t) \hbox{ and } \tilde{R}_i^{\max}(t)\not=R^{\max}_i\right\}\right| }{\left|{\cal J}(t)\right|}\\
&\leq& \frac{\left|\left\{i: \tilde{R}_i^{\max}(t)\not=R^{\max}_i\right\}\right| }{\left|\left\{i: R_i(t)=\tilde{R}_i^{\max}(t)\right\}\right|}.
\end{eqnarray*}

Assume that $N$ short-lived flows are in the network at time $t-D$ and denote by $\cal N$ the set of these short-lived flows. Our proof contains the following two steps:

{\bf Step 1:} We first obtain an upper bound on $$N_1\triangleq \left|\left\{i: \tilde{R}_i^{\max}(t)\not=R^{\max}_i\hbox{ and } i\in{\cal N} \right\}\right|.$$ Considering a short-lived flow (flow $i$) which is in the network at time $t-D,$ we have
\begin{eqnarray*}
\Pr\left(\tilde{R}^{\max}_i(t)\not={R}^{\max}_i\right)\leq (1-p_s^{\max})^D\triangleq\epsilon.
\end{eqnarray*}
Thus, $\Ex[N_1]\leq \epsilon N.$ According to the Chernoff bound, we have
\begin{eqnarray*}
&&\Pr\left(N_1\geq 1.1\epsilon N+D \right)\\
&\leq& \exp\left(-\frac{(1.1\epsilon N+D-\Ex[N_1])^2}{3\Ex[N_1]}\right)\\
&\leq&\exp\left(-\frac{(0.1\epsilon N+D)^2}{3\epsilon N}\right)\\
&\leq&\exp\left(-0.003\epsilon N-0.06D\right).
\end{eqnarray*}

Next note that at most $K\lambda^{\max}$ short-lived flows join the network during each time slot, so we can conclude that
\begin{eqnarray*}
&\Pr\left(\left|\left\{i: \tilde{R}_i^{\max}(t)\not=R^{\max}_i\right\}\right|\geq 1.1\epsilon N +K\lambda^{\max} D+D\right)\\
&\leq \exp\left(-0.003\epsilon N-0.06D\right).
\end{eqnarray*}

{\bf Step 2:} Since at most one flow can be completely transmitted in one time slot, so  least $N-D$ flows are in the network at time $t,$ each having a probability at least $p_s^{\max}$ to be in the best channel state.
\begin{eqnarray*}
&&\Pr\left(\left|\left\{i: R_i(t)=\tilde{R}_i^{\max}(t)\right\}\right|\leq 0.9p_s^{\max} (N-D)\right)\\
&\leq& \Pr\left(\left|\left\{i: R_i(t)={R}_i^{\max}(t)\right\}\right|\leq 0.9p_s^{\max} (N-D)\right)\\
&\leq& \exp\left(-0.003p_s^{\max}(N-D)\right).
\end{eqnarray*}

{\bf Summary:} From step 1 and step 2, we can conclude that
\begin{eqnarray*}
&&\Pr({\cal E}_{\scriptsize miss}(t))\\
&\leq& \frac{1.1\epsilon N+K\lambda^{\max}D+D}{0.9 p^{\max}_s (N-D)}+\exp\left(-0.003p_s^{\max}(N-D)\right)\\
&&\hspace{0.5in}+\exp\left(-0.003\epsilon N-0.06D\right),
\end{eqnarray*} which converges to zero as both $D$ and $N/D$ go to infinity. The proposition holds because the sizes of short-lived flows are bounded and a large workload implies a large number of short-lived flows.
\end{proof}

\section*{Appendix D: Oldest-first Tie-breaking Rule}
%First recall the definition of residue-size-based tie-breaking.
%
%{\bf Smallest-first Tie-breaking:} When the base station decides to select a short-lived flow, it selects flow $i$ having  the smallest $Q_i(t)$ among the set of short-lived flows that have been in the network for at least $D$ time slots, and satisfying $R_j(t)=\tilde{R}_j^{\max}(t)$ or $R_j(t)>Q_j(t).$ If there are multiple such flows, the ties are broken uniformly and randomly. If no such flow exists, ties are broken arbitrarily.

\begin{prop}
The oldest-first tie-breaking is a {\em good} tie-breaking rule.
\end{prop}
\begin{proof}
We assume that at time slot $t-D,$ there are $N>2D^2$ short-lived flows in the network. We group \emph{short-lived flows} into groups $\cal G$ according to the time they arrived at the network such that group ${\cal G}_{\bar{\tau}}(t)$ contains all flows arriving no less than $\bar{\tau}$ time slots ago at time $t,$ and group ${\cal G}_{\tau}(t)$ contains the flows arriving exact $\tau$ time slots ago at time $t$ ($\tau<\bar{\tau}$).

%Now we search the groups starting from group $1$ (the first group) and stop at group $h^*$ if $2D^2>\sum_{h=1}^{h^*}|{\cal G}_h|\geq D^2$ or $|{\cal G}_{h^*}|\geq D^2.$

%These flows are ordered according to the time they arrived so that flow $1$ is the flow that has been in the network for the longest. Note that for any of these flows, we have

{\bf Case 1: } Assume that $|{\cal G}_{\bar{\tau}}(t-D)|\geq D^2.$ We first consider the following probability
\begin{align*}
\Pr\left(\hbox{a flow $\in{\cal G}_{\bar{\tau}}(t)$ is selected at $t$ and  } \tilde{R}^{\max}_i(t)\not=R_i^{\max}\right).
\end{align*} Note that ${\cal G}_{\bar{\tau}}(t)$ can contain at most $K\lambda^{\max}D$ additional flows compared to ${\cal G}_{\bar{\tau}}(t-D)$ since $|{\cal G}_{\tau}|\leq K\lambda^{\max}$ for all $\tau<\bar{\tau}.$  Following the analysis for the uniform tie-breaking in Appendix C, we can easily prove that
\begin{align*}
&\Pr\left(\hbox{a flow $\in{\cal G}_{\bar{\tau}}(t)$ is selected at $t$ and  } \tilde{R}^{\max}_i(t)\not=R_i^{\max}\right)\\
&\rightarrow 0
\end{align*} as $D$ goes to infinity.

Next, note that at most one short-lived flow can be completely transmitted in one time slot, so  ${\cal G}_{\bar{\tau}}(t)$ containing at least $D^2-D$ flows at time $t,$ which implies that
\begin{align*}
&\Pr\left(\hbox{a flow $\not\in{\cal G}_{\bar{\tau}}(t)$ is selected at $t$ and } \tilde{R}^{\max}_i(t)\not=R_i^{\max}\right)\\
\leq&\left(1-p_s^{\max}\right)^{D^2-D}.
\end{align*}
Therefore, we conclude that
\begin{align*}
&\Pr\left(\hbox{the selected flow at $t$ has } \tilde{R}^{\max}_i(t)\not=R_i^{\max}\right)\\
=&\Pr\left(\hbox{a flow ${\cal G}_{\bar{\tau}}(t)$ is selected at $t$ and } \tilde{R}^{\max}_i(t)\not=R_i^{\max}\right)\\
+&\Pr\left(\hbox{a flow $\not\in{\cal G}_{\bar{\tau}}(t)$ is selected at $t$ and } \tilde{R}^{\max}_i(t)\not=R_i^{\max}\right)
\end{align*} which converges to zero as $D$ goes to infinity.

{\bf Case 2: } Assume that $|{\cal G}_{\bar{\tau}}(t-D)|< D^2.$ In this case, we search the groups starting from group ${\cal G}_{\bar{\tau}}(t)$ and stop at group $\tau^*$ if $D^2+DK\lambda^{\max}>\sum_{\tau=\tau^*}^{\bar{\tau}}|{\cal G}_\tau(t)|\geq D^2.$ Note that when $D$ is sufficiently large,  such $\tau^*$ exists and $\tau^*>D$ because $N>2D^2$ and  $|{\cal G}_{\tau}|\leq K\lambda^{\max}$ for all $\tau<\bar{\tau}.$ Considering a certain flow $i$ such that $i\in \cup_{\tau=\tau^*}^{\bar{\tau}} {\cal G}_\tau(t),$ we have that
\begin{align*}
&\Pr\left(\hbox{flow $i$ is selected at $t$ and } \tilde{R}^{\max}_i(t)\not=R_i^{\max}\right)\\
&\leq \Pr\left(\tilde{R}^{\max}_i(t)\not=R_i^{\max}\right)\\
&\leq (1-p_s^{\max})^D\triangleq \epsilon.
\end{align*}   Thus, we can obtain that
\begin{align*}
&\Pr\left(\hbox{a flow $\in \cup_{\tau=\tau^*}^{\bar{\tau}} {\cal G}_\tau(t)$ is selected at $t$ and } \tilde{R}^{\max}_i(t)\not=R_i^{\max}\right)\\
\leq& (D^2+DK\lambda^{\max}) \epsilon,
\end{align*} which converges to zero as $D$ goes to infinity. Further, similar to the analysis in Case 1, we can obtain that when $D$ is sufficiently large,
\begin{align*}
&\Pr\left(\hbox{a flow $\not\in\cup_{\tau=0}^{\tau^*-1} {\cal G}_h(t)$ is selected at $t$ and } \tilde{R}^{\max}_i(t)\not=R_i^{\max}\right)\\
\leq&\left(1-p_s^{\max}\right)^{D^2},
\end{align*}  which converges to zero as well.

Combining Case 1 and 2, we can conclude that the oldest-first tie-breaking is a good tie-breaking rule.
\end{proof}

%\bibliographystyle{IEEEtran}
%\bibliography{U:/Reference/Ying_Lei_Reference}


\end{document}